\documentclass{article}

    \usepackage[nonatbib, final]{neurips_2022}

\usepackage[utf8]{inputenc} 
\usepackage[T1]{fontenc}    
\usepackage{hyperref}       
\usepackage{url}            
\usepackage{booktabs}       
\usepackage{amsfonts}       
\usepackage{nicefrac}       
\usepackage{microtype}      
\usepackage{xcolor}         

\usepackage{enumitem}
\usepackage{amssymb}
\usepackage{amsmath}
\usepackage{amsthm}
\usepackage{graphicx}
\usepackage{caption}
\usepackage{subcaption}
\usepackage{bm}
\usepackage{algpseudocode}
\usepackage{varwidth}
\usepackage{placeins}
\usepackage{mathtools} 
\usepackage{verbatim}
\newcount\Comments
\usepackage{color}
\definecolor{darkgreen}{rgb}{0,0.5,0}
\newcommand{\kibitz}[2]{\ifnum\Comments=1{\color{#1}{#2}}\fi}

\newcommand{\ignore}[1]{}
\newtheorem{theorem}{Theorem}
\newtheorem*{theorem*}{Theorem} 
\newtheorem{lemma}{Lemma}

\newtheorem{proposition}{Proposition}
\newtheorem{observation}{Observation}

\newtheorem{definition}{Definition}

\algtext*{EndFor}
\algtext*{EndIf}
\DeclareMathOperator*{\argmax}{argmax} 

\title{How and Why to Manipulate Your Own Agent:\\ On the Incentives of Users of Learning Agents}

\author{%
  Yoav Kolumbus\\
  The Hebrew University of Jerusalem\\
  \texttt{yoav.kolumbus@mail.huji.ac.il} \\
	\And
	Noam Nisan\\
  The Hebrew University of Jerusalem\\
  \texttt{noam@cs.huji.ac.il} \\
}

\begin{document}

\maketitle

\begin{abstract}
The usage of automated learning agents is becoming increasingly prevalent in many online economic applications such as online auctions and automated trading. Motivated by such applications, this paper is dedicated to fundamental modeling and analysis of the strategic situations that the \emph{users} of automated learning agents are facing. We consider strategic settings where several users engage in a repeated online interaction, assisted by regret-minimizing learning agents that repeatedly play a ``game'' on their behalf. We propose to view the outcomes of the agents' dynamics as inducing a ``meta-game'' between the users. Our main focus is on whether users can benefit in this meta-game from ``manipulating'' their own agents by misreporting their parameters to them. We define a general framework to model and analyze these strategic interactions between users of learning agents for general games and analyze the equilibria induced between the users in three classes of games. We show that, generally, users have incentives to misreport their parameters to their own agents, and that such strategic user behavior can lead to very different outcomes than those anticipated by standard analysis. 
\end{abstract}


\bibliographystyle{splncs03}


\section{Introduction}

This paper deals with the following common type of scenario: several users engage in some strategic online interaction, where each of them is assisted by a learning agent. A typical example is advertisers that compete for advertising slots over some platform. Typically, each of these advertisers enters his key parameters into some advertiser-facing website, and then this website's ``agent'' participates on the advertiser's behalf in a sequence of auctions for ad slots.  Often, the platform designer provides this agent as its advertiser-facing user interface.  In cases where the platform's agent does not optimize sufficiently well for the advertiser (but rather, say, for the auctioneer), one would expect some other company to provide a better (for the advertiser) agent.

The basic model for such scenarios is as follows. There is an underlying $n$-person game that $n$ ``software learning agents'' will play repeatedly. Some of the parameters of this game, typically the players' utilities, are private to the $n$ ``human users,'' and each of these users enters his key parameters into his own software agent. From this point on, these software agents repeatedly play the game on behalf of their users, where each agent aims to optimize the utility of its owner according to the parameters reported to it by using some learning algorithm. A typical learning algorithm of this type will have at its core some regret-minimization algorithm \cite{blum2007learning,hart2000simple}, such as ``multiplicative weights'' \cite{arora2012multiplicative}, ``online gradient descent'' \cite{zinkevich2003online}, or some variant of fictitious play \cite{brown1951iterative,robinson1951iterative}, such as ``follow the perturbed leader'' \cite{hannan1957lapproximation,kalai2005efficient}.

On an intuitive level, these agents, playing between themselves, each aiming to maximize its owner's utility, will reach some kind of a game-theoretic equilibrium of the underlying game, and the average utilities of the agents over time will converge to the utilities of this equilibrium. Specifically, while it is known that the dynamics of regret-minimization algorithms may fail to converge to any equilibrium \cite{bailey2021stochastic,bailey2018multiplicative,hart2013simple,papadimitriou2018nash}, it is also known that the empirical play statistics of no-regret dynamics do approach what is sometimes called a {\em coarse correlated equilibrium} \cite{blum2007learning,hart2000simple,young2004strategic} (CCE). The notion of coarse correlated equilibrium is a generalized weaker notion than the Nash equilibrium,\footnote{For games in which there are at most two actions for each player this notion is equivalent to Aumann's correlated equilibrium \cite{aumann1974subjectivity,aumann1987correlated}, but for general games it is a weaker and more general notion.} and is essentially a formalization of the no-regret property for each agent. For formal definitions and further discussion and analysis of regret-minimization dynamics, see Appendix \ref{sec:convergence}.

While the dynamics of learning agents in repeated games and their convergence properties have been the focus of extensive studies since the early days of game theory \cite{blackwell1956analog,blackwell1954controlled,brown1951iterative,hannan1957lapproximation,robinson1951iterative,shapley1964some} and in later works, e.g., \cite{blum2007learning,daskalakis2021near,foster1997calibrated,freund1999adaptive,fudenberg1999conditional,hart2000simple,syrgkanis2015fast}, our focus is on a different question of analyzing the \emph{incentives of the users} of such agents. In our scenario where (human) users report some of their parameters to their (software) agents, the learning process of the agents depends on the parameters that the users report to them. This motivates the following question: \textbf{Given a game and a set of learning agents (one for each user), what parameters should the users report to their own learning agents?} 

Since in the reached equilibrium every agent is best-replying to the others' behavior, it would seem obvious that each user would maximize his utility by indeed {\em reporting his true parameters to his own agent}, thus indeed allowing the agent to optimize for the true utilities. However, as we show, on closer inspection, this is not necessarily true: while a user is guaranteed that his agent best-replies to the others' empirical play (in the sense of having low regret), this empirical play itself of the other agents and the outcome of the joint dynamics are dependent on the behavior of our user's own agent. Thus, the interactions between the agents induce a ``meta-game'' between the users, in which the users' actions are the parameters that they report to their own agents, and the users' utilities are determined by the long-term empirical average of the agents' joint dynamics. 

This is somewhat similar to the situation in classic repeated games: since agents respond to the previous actions of the other players, the ``dynamics'' do not necessarily give an equilibrium of the underlying single-shot game.  In fact, the folk theorem characterizes a wide set of equilibria of the repeated game that can be reached using the correct combination of punishments and rewards that each player uses to affect the others' behavior (see, e.g., \cite{benoit1984finitely,kalai2010commitment,osborne1994course,segal2007tit}). In some sense, the critical aspect here is a ``theory of mind'' that players have about each other, 
which enables them to understand how punishments and rewards affect the other players' play in the equilibrium of the repeated game.  

By contrast, in our case, regret-minimizing agents have no ``theory of mind'' since they are each blindly following their own regret-minimization strategy.  
In particular, a user's agent never aims to punish or reward another agent.  A user that aims to influence the outcome of the repeated play between regret-minimizing agents must in some sense have a ``theory of mind'' of regret-minimizing agents, both of his own agent as well
as of the other agents.  

In this paper we establish the first steps in the theoretical modeling and analysis of the incentives of users of such automated learning agents in online strategic  systems. We analyze three classes of games: dominance-solvable games, Cournot competition games, and opposing-interests games, and demonstrate how in all these different settings, users generally have incentives to manipulate their own learning agents by misreporting their parameters to them, and how equilibria of the users' game can have different properties and outcomes than those of the original underlying game. 
   
Our results underline the importance of considering and analyzing \emph{user incentives}: when a strategic system is accessible to its users through learning agents, its original properties are not necessarily preserved, and even strong notions like strict domination may no longer hold, and so a meta-game analysis is required in order to understand or anticipate the actual outcomes. Our results focus on demonstrating the types of phenomena that can happen due to strategic behavior of users of learning algorithms and on showcasing them in games that are as transparent for analysis and as simple as possible. The framework that we propose is general for any learning dynamics and can be studied in any game -- either analytically or by using simulations. In a companion paper \cite{kolumbus2021auctions} we use the model that we propose here to analyze repeated auctions with regret-minimizing agents, and show that these phenomena indeed occur also in the auction setting. That is, counterintuitively, in the (non-truthful) first-price auction users prefer to submit truthful reports to their agents, while in the (dominant-strategy truthful) second-price auction users have incentives to manipulate their own agents by submitting non-truthful reports of their valuations.

\section{The Meta-Game Model} 

We define the ``meta-game'' between ``users'' (``players'') in terms of the outcomes of the repeated games that their agents play on their behalf.  
In our definition, all that we formally need from these agents is that they repeatedly interact with each other, where at each point
in time each agent's algorithm has observed the past play of all agents in all previous time steps,
and then needs to determine its own next action based on this information (and the parameters given by its user).  Intuitively,
however, we are thinking of agents that aim to maximize some utility for their user, specifically, regret-minimizing agents.  

\begin{definition} 
A {\em user--agent meta-game} has the following ingredients. 
\vspace{-2pt}
\begin{itemize}[leftmargin=14pt]
\setlength\itemsep{2pt}
	\item {\bf Users and Agents}: We have $n$ ``human'' users, each with his fixed ``software'' learning agent. 
	\item {\bf Agent Strategy Spaces} $A_1,...,A_n$: The agents ``play'' against each other for a period of $T$ time steps, where in each step $t=1,...,T$ each agent 
	$i$ plays an action
	$a_i^t \in A_i$ and then gets as feedback the actual play of all the agents $\textbf{a}^t=(a_1^t,...,a_n^t)$.
	\item {\bf User Parameter Spaces} $P_1,...,P_n$: Each user $i$ inputs a ``declaration'' $p_i$$\in$$P_i$ into his agent.
	\item {\bf Agent Utility Functions} $u_1,...,u_n$: The utility function of agent $i$ is 
	$u_i:P_i \times A_1 \times \cdots \times A_n \rightarrow \mathcal{R}$.  At each time step $t$, each agent $i$ gets a utility of $u_i(p_i, \textbf{a}^t)$,
	where $p_i$ are the parameters given to it by its user.  Our agents aim to maximize this utility.
	\item {\bf User True Types} $s_1,...,s_n$: Each user has a ``true'' type $s_i \in P_i$ that describes the true utility of the user. The
	utility of user $i$ when his true type is $s_i$, and when each user $j$ declares $p_j \in P_j$ to his agent, is
	$U_i(s_i,p_1,...,p_n)=(${\small$\sum_{t=1}^T$}$u_i(s_i,\textbf{a}^t))/T$, where for each $t=1,...,T$, $\textbf{a}^t$ is the vector of actions played by the 
	agents at time step $t$, as defined above. 
	If the agents are randomized, then this expression is actually a random variable, and we 
	define $U_i$ to be its expectation. 
	\item {\bf The User Meta-Game}: Fixing the true parameters $(s_1,...,s_n)$, we get an $n$-person game between the users, 
	where user $i$'s strategy space is $P_i$ and his utility is $U_i(s_i, \cdot)$.  
\end{itemize}
\end{definition}
 
This model is formally defined for any game, fixed tuple of learning algorithms, and fixed ``horizon'' $T$, and for any such the ensuing meta-game can be directly studied by simulations. In order to proceed and theoretically study this meta-game as the horizon $T$ goes to infinity, we would require a setting where it is possible to theoretically analyze the utilities of the resulting dynamics both for the true types (parameters) $s_i$ and for possible deviations $p_i \ne s_i$.  The theoretical analysis in the following sections concerns all regret-minimization algorithms, in the ``limit'' $T \rightarrow \infty$, as we will study cases where the agents' game converges to a known CCE. For further discussion on the convergence of regret-minimizing agents, see Appendix \ref{sec:convergence}. 
We note that while our focus is on regret-minimization dynamics, all our results hold also under an alternative interpretation of our model in which the agents reach a CCE of the game with the declared parameters via an arbitrary black-box device.     

For the games studied, we will start with basic questions like what is the best reply of a player to another player in the meta-game, and then proceed to more advanced questions, specifically, in which cases is the truth $p_i=s_i$ a best reply in the meta-game, and when this is not the case, what is the Nash equilibrium of the meta-game. All proofs in this paper are deferred to the appendix.

\begin{definition}\label{def:manipulation-free}
A user--agent meta-game is called \emph{manipulation-free} if the truth-telling declaration profile, i.e., all users declaring $p_i=s_i$, is a Nash equilibrium of the meta-game.\footnote{This definition is for a fixed game that is defined by the true parameters $s_1,...,s_n$ (but in the context of given parameter spaces $P_1,...,P_n$). One may also naturally look at the family of games for all possible $s_1,...,s_n$ and discuss truthfulness in the sense used in mechanism design \cite{nisan2007introduction}, but we leave this for further follow-up work.}
\end{definition}

\section{Dominance-Solvable Games}\label{sec:DS-games}
The first class of games that we consider are dominance-solvable games. A main interest in the literature in studying these types of games has been as design objectives due to their stable  strategic structure, e.g., for voting mechanisms \cite{dhillon2004plurality,moulin1979dominance} and contract design \cite{babaioff2022optimal,halac2020raising,segal1999contracting}. Formally, games of this class have a unique pure Nash equilibrium that is also the single CCE, and thus we know that the empirical play statistics of any regret-minimization dynamics will provide, in the limit, the utilities of this equilibrium.
For completeness, we give here the definition of a dominance-solvable game. 
\begin{definition} \label{dominance-solvable game}
A game is called \emph{dominance solvable} if there exists an order of iterated elimination of strictly dominated strategies that leads to a single strategy profile (the unique Nash equilibrium). 
\end{definition} 

We show that even in these strategically simple games, users of learning agents face non-trivial strategic considerations, and, specifically, even users who have a dominant strategy in the game may obtain further gains (beyond their dominant-strategy outcome) by manipulating their own agents.     

We begin by demonstrating our agenda on a simple $2$$\times$$2$ two-person game where one of the players has a dominant strategy. The other player then has a strict best reply, and so the game is dominance-solvable. It is not difficult to see that any dynamics of regret-minimizing learning agents playing such a game will converge to the pure equilibrium, since the agent with the dominant strategy will learn to play only its dominating strategy, regardless of the actions of the second agent, and then the second agent will learn to best-reply to that. Formally, the time average of regret-minimization dynamics must converge to the (unique) Nash equilibrium as it is also the unique CCE. For further details, see Appendix \ref{sec:appendix-DS-games}.

This simple analysis also shows that the other player has no profitable manipulation: since the agent of the player with the dominant strategy can indeed learn to play it whatever the other player does, the other player can do no better than to best-reply to the dominant strategy.    

It may also seem intuitive that the player with the dominant strategy can have no profitable manipulation either, but this turns out to be false.  Consider the game depicted in Figure \ref{fig:game-matrices} (left), in which the row player has a dominant strategy to play the bottom row. The unique pure Nash equilibrium of this game gives the row player utility $u_1=2$. We now get to our point where the utility of each player is private to him (at least partially). In our example, suppose that the constants $1$ and $3$ in the game description are private to the two users, respectively, and that each learning agent gets the value of the parameter ($c$ and $d$, respectively) from their users.  After each agent gets its own parameter, the two agents then engage in repeated play. In this play, each agent minimizes regret for its owner {\em in the game with the parameters that were given to it} (rather than the true parameters that are known only to the user), as shown in Figure \ref{fig:game-matrices}  (right). If the row player declares instead of the true $c=1$ a value of, say, $c=5$, then the (declared) game has a unique mixed Nash equilibrium, which is also its unique CCE \cite{calvo2006set,moulin1978strategically}, where the row player plays the top row with probability $p=1/2$, and the column player plays the left column with probability $q=1/4$, giving the row player a utility of $u_1=3$.  
Declaring an even higher value $c \rightarrow \infty$ will decrease $q \rightarrow 0$, leading to a higher utility of $u_1 \rightarrow 3.5$.  

Once the row player has manipulated his input, the column player may also beneficially do so. The meta-game does not literally have a Nash equilibrium (as the strategy spaces of the players are a continuum and no continuity of utility is guaranteed), but it has an $\epsilon$-equilibrium\footnote{In an $\epsilon$-equilibrium no player can gain more than $\epsilon$ by deviating.} for any $\epsilon>0$: as $c \rightarrow \infty$ and $d=4-\delta$ with $\delta>0$ and $\delta \rightarrow 0$. In this case, $p \rightarrow 1$ and $q \rightarrow 0$, leading to an $\epsilon$-equilibrium with utilities $u_1 \rightarrow 3$, $u_2 \rightarrow 4$.

\begin{figure}[!t]
\centering
		\includegraphics[width=0.5\linewidth]{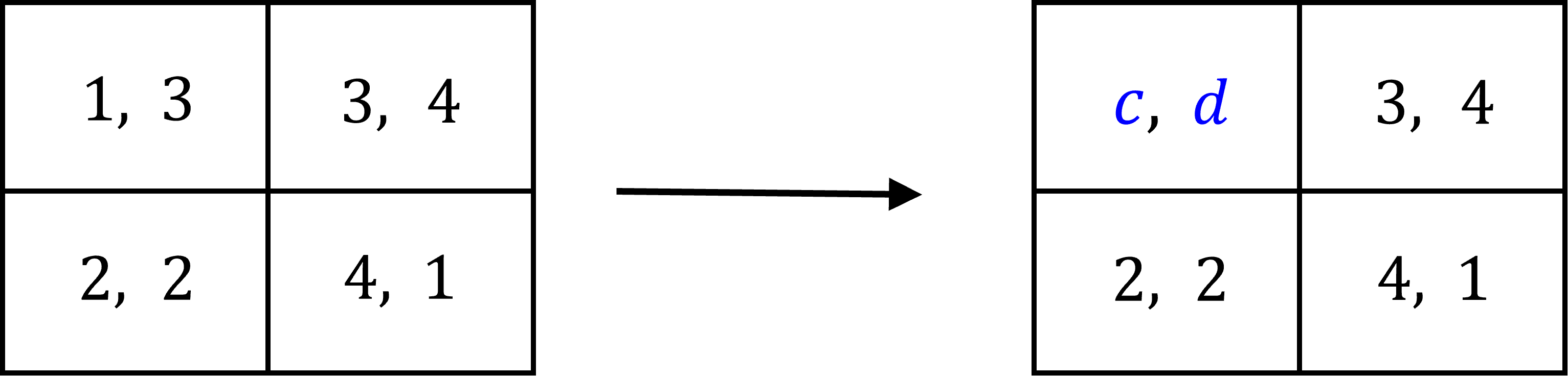}
		\caption{An example of parameter manipulation in a dominance-solvable game. Left: the payoff matrix of the true game. Right: the manipulated payoffs that the players provide to their learning agents. The row player selects the parameter $c$ and the column player selects the parameter $d$. }  
		\label{fig:game-matrices}
\vspace{-4pt}
\end{figure}

Notice that the outcome in this $\epsilon$-equilibrium of the meta-game strictly Pareto-dominates the Nash equilibrium of the original game, and thus we may say that the players managed to reach a cooperative outcome. The logic behind this cooperation is that the player with the dominant strategy is given the opportunity to take the point of view of a Stackelberg game where he goes first and the other player best-replies to his strategy. If the dominated strategy is the preferable strategy in the Stackelberg game, as it is in our example, then a manipulation can approach it. 

This example demonstrates how even in very simple games, users may have incentives to misreport their parameters to their agents, and that even a strong notion like strict domination does not guarantee the stability of truthful declarations. This phenomenon is in fact general for a large class of games:
 
\begin{theorem}\label{thm:DS-game-manipulations}
In any $n \times m$ game where one of the players has a dominant strategy, if the Stackelberg outcome of the game is different from the truthful Nash equilibrium outcome, then, for a sufficiently wide user parameter space, the game is not manipulation-free, and, specifically, the player with the dominant strategy has an incentive to manipulate his declaration.
\end{theorem}
For the special case of $2$$\times$$2$ games this result holds even when users can manipulate only a single one of their parameters, and for specific subclasses of games it is possible to characterize the equilibria of the meta-game, as in our example above. Additionally, Theorem \ref{thm:DS-game-manipulations} is in fact even more general and applies not only to games with a dominant strategy, but to any dominance-solvable $n \times m$ game where some player has a Stackelberg value that is higher than his utility in the truthful Nash equilibrium.

\section{Cournot Competition Games}

The second class of games that we consider are the classic Cournot competition games \cite{cournot1838recherches,even2009convergence,mas1995microeconomic} with linear demand functions and linear production costs. These games are contained in a class of games called ``socially concave'' that were identified by \cite{even2009convergence}, who showed that for games of this class the time-average distribution of any regret-minimization dynamics converges to the unique Nash equilibrium of the game, and so we can confidently analyze the utilities obtained in the meta-game.

We consider a game between two firms that are competing for buyers by controlling the quantity that each of them produces. There is a demand function that specifies the market price for any given total quantity produced. In our case we assume that the demand function is linear; i.e., if the two firms produce quantities $q_1$ and $q_2$ respectively then the market price will be $a-b\cdot (q_1+q_2)$, where $a$ and $b$ are commonly known positive constants. The private parameter that each firm will have is its production costs, which we also assume are linear; i.e., firm $i$'s cost to produce quantity $q_i$ is exactly $c_i \cdot q_i$, where $0 \le c_i \le a$ is privately known to firm $i$. The utility of firm $1$ is given by $u_1(q_1,q_2)= q_1 \cdot (a-b\cdot (q_1+q_2) - c_1)$; similarly, $u_2(q_1,q_2)= q_2 \cdot (a-b\cdot (q_1+q_2) - c_2)$.

The Nash equilibrium of the game depends on the parameters as follows. 
If $a + c_2 - 2c_1 > 0$ and $a + c_1 - 2c_2>0$, then the Nash equilibrium is  
{\small
$q_1 = \frac{1}{3b}(a + c_2 - 2c_1)$
}
{\small  
and   
$q_2 = \frac{1}{3b}(a + c_1 - 2c_2)$. 
}
If $a + c_1 - 2c_2>0$ and $c_1 < a$, the Nash equilibrium is 
{\small
$q_1 = \frac{a-c_1}{2b}$
}
{\small  
and
$q_2 = 0$,   
}
and symmetrically, if $a + c_2 - 2c_1>0$ and $c_2 < a$, the equilibrium is $q_1 = 0$ and $q_2 = \frac{a-c_2}{2b}$. 
Otherwise, in the Nash equilibrium both players produce zero.

Thus, there are four parameter regions of interest associated with the four possible types of unique Nash equilibria of the game, as illustrated in Figure \ref{fig:cournot_regions}. The parameter region $A=\{c_1,c_2|a + c_2 - 2c_1 > 0, \ a + c_1 - 2c_2 > 0, \ c_1>0, \ c_2>0\}$ is the region where both agents produce positive quantities (the shaded areas in region $A$ in the figure relate to equilibria of the meta-game, as explained below). The parameter regions $B=\{c_1,c_2|a + c_1 - 2c_2 > 0, \  0 < c_1 < a, \ c_2>0\}$ and $C=\{c_1,c_2|a + c_2 - 2c_1 > 0, \  0 < c_2 < a, \ c_1>0\}$ are regions where only one player produces a positive quantity. In the remaining region, region $D=\{c_1,c_2|c_1, c_2 \geq a\}$, both agents produce zero. 

\begin{figure}[!t]
	\centering
		\includegraphics[width=0.4\textwidth]{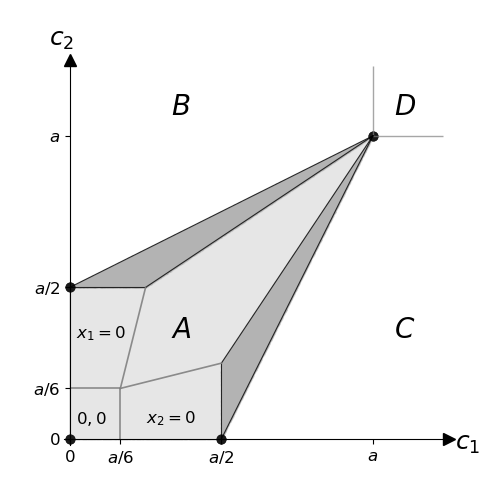}
	\caption{Parameter declaration regions in Cournot competition games.}
	\label{fig:cournot_regions}
\vspace{-6pt}
\end{figure}

As a running example, we will consider the case where $a=b=1$ and $c_1=c_2=1/2$, for which the standard analysis yields that in equilibrium each player produces the quantity $q_i=1/6$, the price is thus $2/3$, and the utility of each player is $1/36$.

We now turn to look at the meta-game in which each player reports his production cost to his own agent, and then the agents repeatedly play the declared game.  
That is, each player $i$ reports a declared cost $0 \le x_i \le a$ in the parameter space and then the agents reach the equilibria
with the {\em declared} costs. It turns out that in our example, firm $1$'s best reply to firm $2$'s true cost is $x_1=3/8$, rather than the
truth $c_1=1/2$, which increases its utility (calculated, of course, according to the true costs) to $1/32 > 1/36$.  
That is, the firm under-represents its production costs to its own agent, causing the agent to over-produce.  While this over-production by itself hurts our user (the firm), the benefit is that our user -- as opposed to its automated learning agent -- understands that this aggressive declaration will
lead to a reduction in the production of the other firm's agent, making up, and more, for the revenue loss from its own over-production.

Analyzing the equilibrium of the meta-game between the users, one obtains that the equilibrium 
is the declaration profile $x_1=x_2=2/5$, rather than the truthful declarations $c_1=c_2=1/2$. That is, the two players under-represent their production costs to their agents, causing them to over-produce, where each of them produces a quantity $q_i=1/5$, which is strictly higher than the production of $1/6$ obtained in the equilibrium of the truthful game. In the equilibrium of the meta-game the utility of each of the two players drops to $1/50$, which is significantly less than the original utility of $1/36$.  We see that the players are locked here in a sort of a prisoners' dilemma where each of them benefits from a unilateral deviation from the truth, but when they both deviate, they both suffer losses. 

Our analysis shows that these comparative statics generalize as long as both players keep producing a non-zero quantity in the meta-game equilibrium.  In some cases, which we explicitly describe below, the player with lower costs can under-represent his costs in a sufficiently extreme way so as to drive the other player completely out of the market. In such cases the quantity produced by the player who remains in the market is still larger than the total quantity produced by both players when they play the truth in the meta-game, but his utility increases. We also completely characterize the rather limited set of cases where the meta-game is manipulation-free. 

\begin{theorem}\label{thm:cournot-meta-game-NE}
(1) In any two-player linear Cournot competition with linear production costs, 
the total quantity produced in the equilibrium of the meta-game is greater than or equal to the total  
quantity produced when the players play the truth in the meta-game, and the price is thus lower. 
(2) If both players continue to produce a non-zero quantity in the meta-game equilibrium, 
then their utilities are less than or equal to their utilities when both play the truth in the meta-game.
If, on the other hand, one of the players produces zero in the meta-game equilibrium, then 
the utility of the producing player is greater than or equal to his utility when both play the truth in the meta-game. 
\end{theorem}

\begin{theorem}\label{thm:cournot-manipulation-free}
In a two-player linear Cournot competition with linear production costs, the meta-game is manipulation-free if and only if either (at least) one of the players produces zero when both play the truth in the meta-game, or both players have zero production costs, $c_i=0$.
\end{theorem}

Figure \ref{fig:cournot_regions} illustrates the parameter ranges in which the different types of equilibria of the meta-game exist. The dark-shaded areas show the parameter ranges where in the equilibrium of the meta-game the player with the low production cost drives the other competitor out of the market, and the light-shaded area shows the parameter range where both players declare costs lower than their true costs and produce positive quantities in the equilibrium of the meta-game. The regions denoted in the figure by $x_1=0$ and $x_2=0$ show where players $1$ and $2$, respectively, declare a cost of zero to their agents in  equilibrium, and the region denoted by $0,0$ is the range where both players declare zero in equilibrium. In the remaining kite-shaped region in region $A$, both players declare positive costs that are less than their true costs. Finally, in regions $B,C$, and $D$, there is no competition and the equilibrium declarations are truthful. For further details, see Appendix \ref{sec:appendix-Cournot}.

\section{Opposing-Interests Games}\label{c2d3}

The next class of games that we analyze are games in which there is a single mixed Nash equilibrium and no pure equilibrium.  
The prototypical example is matching pennies. In these games, often called fully mixed games, the unique Nash equilibrium is also the unique coarse correlated equilibrium 
\cite{calvo2006set,moulin1978strategically}, and thus the dynamics of regret-minimizing agents will approach this equilibrium. We focus on a subclass called {\em opposing-interests} games, where the first player gets higher utilities along the main diagonal than he gets along the other diagonal, and the opposite is true for the second player. This subclass of fully mixed games includes many games that are similar to matching pennies, and specifically includes all constant-sum games that are fully mixed.

We begin with an example of the following variant of matching pennies, where one of the utilities for each player is changed from the standard value of 1
to another value, as shown in Figure \ref{fig:fully_mixed_manipulation_game_matrix} (left).

A standard analysis shows that the single mixed Nash equilibrium of this game is where the row player plays the top row with probability $p=2/3$ 
(and plays the bottom row with probability $1-p=1/3$), and the column player plays the left column with probability $q=2/5$ (and the right column with
probability $1-q=3/5$).  Calculating the utilities of the two players in this equilibrium gives us $u_1=1/5$ for the row player, and $u_2=1/3$ for the column player. Running a (typical) simulation of multiplicative-weights learning agents that repeatedly play this game against each other, we observe the dynamics shown in Figure \ref{fig:fully_mixed_p_q_dynamics}. 
As is well known \cite{bailey2018multiplicative,hart2013simple,papadimitriou2018nash}, and as we can clearly see, there is no convergence in the behavior of the agents. However, if we write down the empirical probabilities of play of each of the four combinations of the players' strategies (as shown in Figure \ref{fig:2x2_distributions_table}) we get (close to) the Nash equilibrium probabilities, as theoretically expected \cite{calvo2006set}.\footnote{There are two reasons for not reaching exactly the Nash equilibrium.  First, as our simulations are only for a finite number of steps and with a finite step size, the multiplicative-weights algorithm does not fully minimize regret but only nearly so, and thus leads only to a near-equilibrium.  Second, as the  algorithm is randomized there is an expected stochastic error. These error terms are theoretically of an order of magnitude of $O(1/${\small$\sqrt{T}$}$)$, where $T$ is the number of rounds, which fits the observed deviations 
in our simulation for $T=50$$,$$000$ rounds.}

\begin{figure}[t]
\centering
  \includegraphics[width=0.5\linewidth]{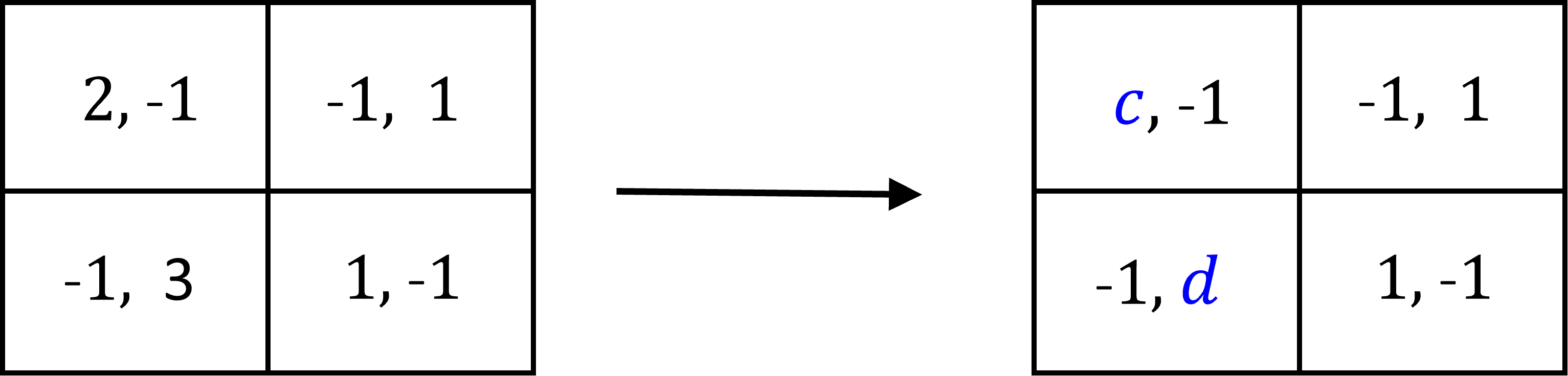}
  \caption{An example of parameter manipulation in an opposing-interests game. Left: the payoff matrix of the true game. Right: the manipulated payoffs that the players provide to their learning agents. The row player selects the parameter $c$ and the column player selects the parameter $d$.
	}
  \label{fig:fully_mixed_manipulation_game_matrix}
\vspace{-2pt}
\end{figure}
 
In our example, suppose that the constants $2$ and $3$ in the game description are parameters that the two users, respectively, declare to their agents, as shown in Figure \ref{fig:fully_mixed_manipulation_game_matrix} (right). Now we ask ourselves, what should the players do in order to maximize their utility?  
What parameter should, say, the row player report to his agent so as to maximize his expected utility over the whole run of the 
learning agents? It would seem natural to assume that entering the true value, in our case $2$, should be the best possible: after all, the agent is optimizing for the value entered into it, and so the row player should give his agent the correct value to optimize for. However, as we have seen also in other types of games, this intuition is again false. Suppose that the row player reports, e.g., $c=1$ as his parameter to his learning agent (and suppose that the column player sticks to the truth, $d=3$).  When the two agents now repeatedly play against each other they reach (in the limit empirical distribution sense) the Nash equilibrium of the game with these values of the parameters, which, as one may calculate, is $p=2/3$, $q=1/2$. The utility of the row player, whose true value is $c=2$, in this resulting distribution is $u_1=1/3$, which is greater than his ``truth-telling'' utility of $1/5$, as we have seen above.\footnote{  
This improvement is even more puzzling when we take a closer look: in the Nash equilibrium of the original game, the row player has mixed between his two pure strategies, implying that they both give him the same utility. Thus, the only way that our row agent can improve his utility is by first causing the column player's agent to change its distribution of play, and only then it could take advantage of this change. The column player's agent's behavior, however, does not depend directly on the row player's utilities. The solution is that there is an indirect dependence that goes through the actual play of the row agent. Taking advantage of this indirect dependence obviously requires understanding the ``theory of mind'' of the column agent's algorithm. In our case, as the row user understands that the agents' dynamics will reach the Nash equilibrium of the {\em declared game}, he can manipulate the parameter that he gives to his agent so as to indirectly cause an increase in $q$, the probability that the column agent plays the first column, an action that the row player finds favorable.
}

Just as the row player can gain by misreporting his utility to his agent, so can the column player. Had the column player declared $d=1$ (now with the row player declaring the truth $c=2$), his utility would have increased to $u_2=2/5$.  Had they both manipulated their ``bids'' and declared $c=1, d=1$ to their respective agents, then both would have benefited relative to telling the truth, getting utilities of $u_1=1/4$ and $u_2=1/2$. 

Next, let us look for an equilibrium of the users' game. We need to determine for which declarations $c,d$ will neither of the two players wishes to unilaterally change his declaration (where the true utilities are still set to $c=2, d=3$ in our example). Our analysis shows that the equilibrium declarations in the meta-game, are $c=3$, $d=1/3$. The equilibrium of the declared game that is played between the agents in this case is $p=2/5$ and $q=1/3$.  
The utilities of the two players (according to the true game parameters) in this distribution are $u_1=1/5$ and $u_2=1/3$ (for further details, see Appendix \ref{sec:appendix-opposing-interests-games}). Surprisingly, these are the same utilities that our players obtained when telling the truth. In fact, this result is not a coincidence. The following theorem characterizes the equilibrium of the meta-game and the utilities obtained in it.
While in our example we fixed a simple parameter space where $c$ was the single parameter of the first player and $d$ the single parameter of the second player, the following results apply to a wide range of parameter spaces where any one of the utilities of each player in the game is a parameter that the player can manipulate. As mentioned, this analysis is for the ``limit'' meta-game and holds for every pair of regret-minimizing agents.

\begin{figure}[!t]
\centering
\vspace{-8pt}
	\begin{subfigure}{0.45\linewidth}
	\vspace{6pt}
		\includegraphics[width=0.96\linewidth]{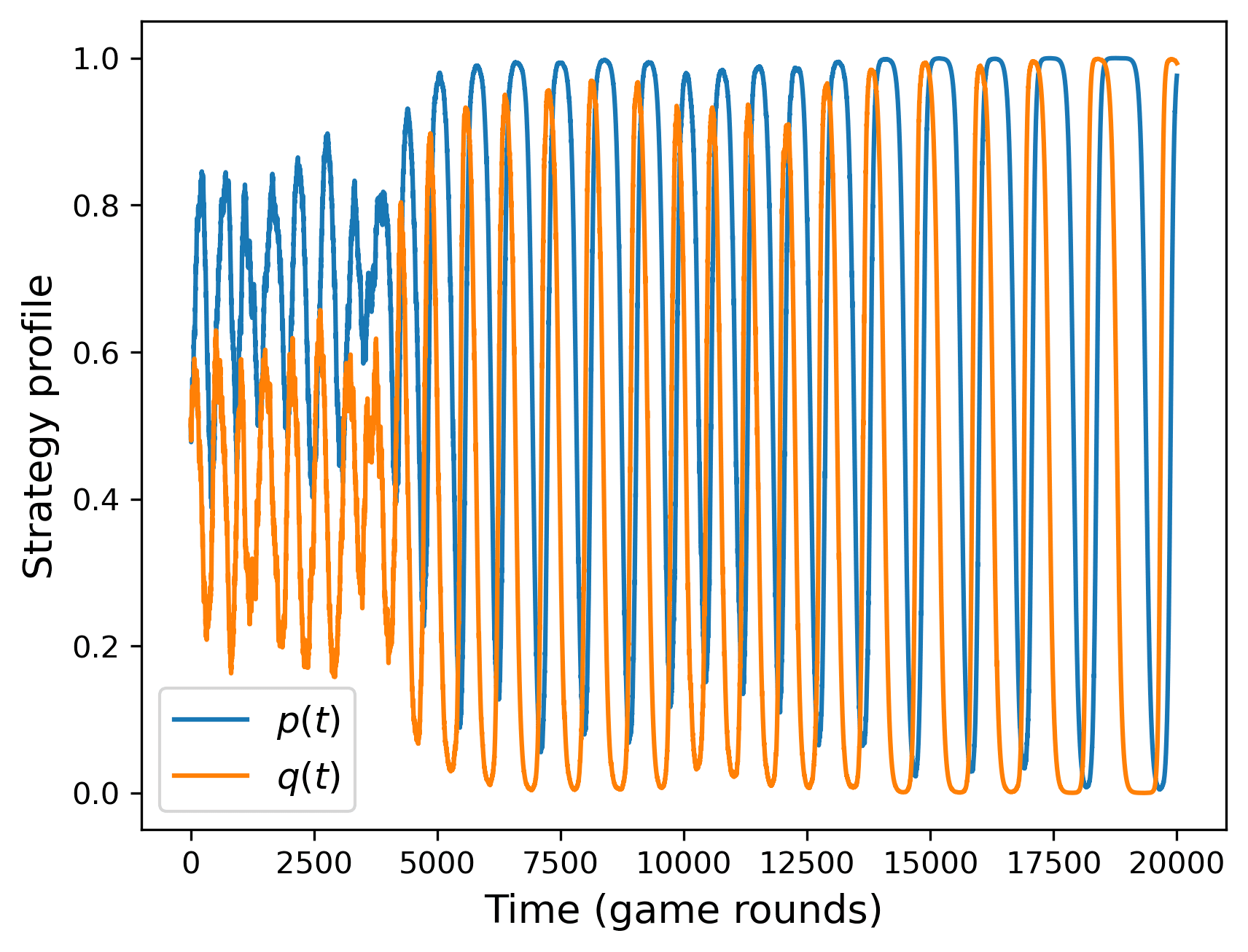}
		\vspace{-5pt}
		\caption{Mixed strategy dynamics}  
		\label{fig:fully_mixed_p_q_dynamics}
	\end{subfigure}
	\hspace{7pt}
	\begin{subfigure}{.52\linewidth} 
	\small
	\begin{center}
									\hspace{-30pt}Nash equilibrium distribution 
	\vspace{-3pt}								
	\end{center}
	\begin{tabular}{c|c|c|}

		\toprule
									& \ \ \ $q = 0.4\quad$\ \ \ \ & \ $1-q = 0.6\quad$  \\
		\midrule
		$p = 0.667$   &   $0.267$  &     $0.133$ \\ 
		\hline
		$1-p = 0.333$ &   $0.4\quad$   &   $0.2\quad$ \\
		\bottomrule
		\vspace{0.1pt}
	\end{tabular}

	\begin{center}
									\hspace{-30pt}Empirical distribution 
	\vspace{-3pt}
	\end{center}
	\begin{tabular}{c|c|c|}

		\toprule
									& \ \ \ $q = 0.405$ \ \ \ & \ $1-q = 0.595$ \\
		\midrule
			$p = 0.665$ &      $0.270$ &       $0.135$ \\
					\hline
		$1-p = 0.335$ &     $0.395$ &         $0.200$ \\
		\bottomrule
	\end{tabular}
	\vspace{9.4pt}
	\caption{Empirical and Nash equilibrium action distributions}
	\label{fig:2x2_distributions_table}
	\end{subfigure}
\caption{Dynamics and payoffs in the game shown in Figure \ref{fig:fully_mixed_manipulation_game_matrix}. Left: dynamics of multiplicative-weights agents in the game with the true parameters $c=2$, $d=3$. Right: the Nash equilibrium distribution and the empirical action distribution from simulations of multiplicative-weights agents.}
\label{fig:fully_mixed_dynamics_and_payoffs_MW}
\end{figure}

\begin{theorem}\label{thm:opposing-interests-NE}
In opposing-interests $2$$\times$$2$ games, the utilities of the two players in an equilibrium of the meta-game are the same as the utilities obtained when both play the truth in the meta-game. 
\end{theorem}
Thus, in these types of games the players do have incentives to unilaterally manipulate their own agents, but in equilibrium they can neither cooperate in the sense of improving their utilities nor do they suffer losses. This is in contrast to what we saw in dominance-solvable games and in Cournot competition games. Additionally, we characterize the cases where these games are manipulation-free: 

\begin{theorem}\label{thm:opposing-interests-manipulation-free}
An opposing interests $2$$\times$$2$ game is manipulation-free if and only if its Nash equilibrium is symmetric under player permutations (i.e., the same equilibrium distribution is obtained when player indices are switched). 
\end{theorem}

\section{Further Related Work}

Regret minimization in repeated strategic interactions and in online decision problems has been extensively studied in the literatures of game theory, machine learning, and optimization. Early regret-minimization algorithms were related to the notion of fictitious play \cite{brown1951iterative,robinson1951iterative} (a.k.a. ``follow the leader''), which in its basic form does not guarantee low regret, but its smoothed variants, such as ``follow the perturbed leader'' (FTPL) \cite{grigoriadis1995sublinear,hannan1957lapproximation,kalai2002geometric,kalai2005efficient} and ``follow the regularized leader'' (FTRL) \cite{shalev2011online}, are known to guarantee an adversarial regret of $O(${\small$\sqrt{T}$}$)$ in $T$ decision periods; for more recent advances along these lines see, e.g., \cite{daskalakis2018last,mertikopoulos2018cycles,syrgkanis2015fast}. Other common approaches to regret-based learning are the ``multiplicative-weights'' algorithm, which has been developed and studied in many variants (see \cite{arora2012multiplicative} and references therein, and see \cite{chen2020hedging,daskalakis2021near,piliouras2021optimal} for results on more advanced variants of this algorithm), and the family of algorithms that are based on the regret-matching approach that has been well studied in several settings (see \cite{hart2013simple} and references therein). For a broad discussion and for further references on regret-minimization dynamics, see \cite{hart2013simple} and \cite{cesa2006prediction}. 

Our work formalizes the meta-game faced by users of regret-minimizing learning algorithms, and asks whether and when users can benefit from manipulating their own agents. To our knowledge, there is no prior analysis or modeling of these strategic interactions between users of learning agents (which are induced by the dynamics of their agents). In a companion paper \cite{kolumbus2021auctions}, we study an application of our meta-game model in auctions, where we analyze the dynamics and outcomes of repeated auctions played between regret-minimizing agents of a class that includes many natural algorithms such as multiplicative weights. We show in that paper that in the meta-game induced between the users of such auto-bidding agents, the second-price auction loses its incentive-compatibility property, while the first-price auction becomes incentive compatible. 

As discussed in the introduction, our work is related to the broad field of equilibria in repeated games \cite{mailath2006repeated,osborne1994course}, but the situation that the users of learning agents in our setting are facing differs from classic repeated games in significant aspects.   
Conceptually closer works are \cite{brafman2002efficient} that study equilibria between policies in the repeated game (i.e., strategies that are conditional on the history of play), and works on program equilibria \cite{lavictoire2014program,oesterheld2019robust,oesterheld2022safe,tennenholtz2004program} in which each agent can read the commitments made by the other agents and condition its actions on these commitments. 
These models, however, are technically very different from our model, and in particular, the notions of equilibrium are different from a meta-game equilibrium. 

In a broader perspective, our work is related to a research area that can be called ``strategic considerations in machine-learning systems,'' with a growing body of work at the intersection of machine learning, algorithmic game theory, and artificial intelligence that addresses this topic from different perspectives, 
including learning from strategic data \cite{cai2015optimum,chen2019learning,dong2018strategic,ghalme2021strategic,haghtalab2020maximizing,hardt2016strategic,levanon2021strategic}, Stackelberg games \cite{birmpas2020optimally,Jiarui2019imitativeFollower}, security games \cite{gan2019manipulating,haghtalab2016three,nguyen2019imitative,nguyen2019deception,shi2020learning}, and recommendation systems \cite{ben2018game,tennenholtz2019rethinking}. More closely related works are \cite{braverman2018selling,deng2019strategizing,mansour2022strategizing}, which deal with optimization against regret-minimizing agents. The possibility of obtaining increased gains when playing against a no-regret algorithm that is studied in these works is conceptually close to our work. A basic difference, however, is that these works consider a single optimizer facing a no-regret algorithm -- a setting that induces optimization problems, rather than games -- whereas we study games between the users that are induced by their learning agents' dynamics, and in which all users can act strategically and reason about the strategies of their peers. Additionally, in the meta-game a user does not need to select the actual step-by-step actions in the repeated game, but instead chooses the declaration to input into the automated agent, whereas the direct optimization over the action space of the underlying game is performed only by the agents.  

Finally, the idea that inputting a ``wrong'' reward function into a learning algorithm can in some cases improve actual outcomes has been long known in the context of reinforcement learning \cite{barto2013intrinsic,kulkarni2016hierarchical,csimcsek2006intrinsic,singh2004intrinsically,singh2010intrinsically}. Specifically, and closer to our context, in Markov games \cite{littman1994markov} it has been shown that the dynamics of reinforcement learners with certain intrinsic reward functions can lead to improved actual utility to all the agents playing the game \cite{austerweil2016other,eccles2019learning,jaques2019social,leibo2017multi}. This literature, however, does not consider  interactions between users of such agents and their incentives when entering their parameters into their learning agents. We view the analysis of our model for Markov games with reinforcement-learning agents as a natural and interesting extension, but we leave this for future work, and focus in the current paper on repeated games with regret-minimizing agents.

\section{Conclusion}

The present study deals with the modeling and analysis of scenarios in which human players use autonomous learning agents to perform strategic interactions with other players on their behalf. The usage of automated learning agents is becoming increasingly common and prominent in many real-world economic systems and online interactions such as online auctions \cite{balseiro2021robust,cesa2014regret,daskalakis2016learning,deng2021autoBidding,kolumbus2021auctions,nekipelov2015econometrics,mehta2022auction,noti2021bid}, financial markets \cite{lei2020time,li2019deep,algotrade2013,ZHOU2019186}, and other systems \cite{rahwan2019machine}. Understanding the impact of this transition to automated agents on strategic systems and on their outcomes is a challenge in itself, and a prerequisite for studying how to better design such automated-interaction systems.

The framework that we propose for analyzing the meta-games between the users, and the results of our analysis, highlight the fact that the strategic nature of interactions does not disappear and does not remain the same when direct human play is replaced by automated agents, but rather the introduction of these learning agents changes in specific ways the rules of the game that the human users (who are the actual stakeholders in the system) are facing. Every automated agent that operates in a system on behalf of its user needs to receive some input from the user (such as preferences, goals, or constraints); we show that at this interaction point between a user and his own learning agent, the user faces a non-trivial strategic decision in which he needs to consider also the decisions of the other users. The goal of our model is to describe these interactions and formalize them. 

Our contributions in this paper include the basic definitions related to these phenomena, demonstrations of how users of learning agents can profitably 
manipulate their own agents in several settings, as well as identification of cases where this cannot be done. These results have implications both for users of 
learning agents in strategic settings and for platform designers who need to take into account these types of manipulations. We believe that the present paper, along with its companion paper \cite{kolumbus2021auctions} where we study these phenomena in online auctions, only scratch the surface of these types of questions, and much work remains to be done.


\section*{Acknowledgments}
This project has received funding from the European Research Council (ERC) under the European Union's Horizon 2020 Research and Innovation Programme (grant agreement no. 740282).


\bibliography{meta_games_arxiv_v2}

\begin{thebibliography}{10}
\providecommand{\url}[1]{\texttt{#1}}
\providecommand{\urlprefix}{URL }

\bibitem{arora2012multiplicative}
Arora, S., Hazan, E., Kale, S.: The multiplicative weights update method: a
  meta-algorithm and applications. Theory of Computing  8(1),  121--164 (2012)

\bibitem{aumann1974subjectivity}
Aumann, R.J.: Subjectivity and correlation in randomized strategies. Journal of
  mathematical Economics  1(1),  67--96 (1974)

\bibitem{aumann1987correlated}
Aumann, R.J.: Correlated equilibrium as an expression of bayesian rationality.
  Econometrica: Journal of the Econometric Society pp. 1--18 (1987)

\bibitem{austerweil2016other}
Austerweil, J.L., Brawner, S., Greenwald, A., Hilliard, E., Ho, M., Littman,
  M.L., MacGlashan, J., Trimbach, C.: How other-regarding preferences can
  promote cooperation in non-zero-sum grid games. In: Proceedings of the AAAI
  Symposium on Challenges and Opportunities in Multiagent Learning for the Real
  World (2016)

\bibitem{babaioff2022optimal}
Babaioff, M., Kolumbus, Y., Winter, E.: Optimal collaterals in multi-enterprise
  investment networks. In: Proceedings of the ACM Web Conference 2022. pp.
  79--89 (2022), \url{https://dl.acm.org/doi/10.1145/3485447.3512053}

\bibitem{bailey2021stochastic}
Bailey, J.P., Nagarajan, S.G., Piliouras, G.: Stochastic multiplicative weights
  updates in zero-sum games. arXiv preprint arXiv:2110.02134  (2021)

\bibitem{bailey2018multiplicative}
Bailey, J.P., Piliouras, G.: Multiplicative weights update in zero-sum games.
  In: Proceedings of the 2018 ACM Conference on Economics and Computation. pp.
  321--338 (2018)

\bibitem{balseiro2021robust}
Balseiro, S., Deng, Y., Mao, J., Mirrokni, V., Zuo, S.: Robust auction design
  in the auto-bidding world. Advances in Neural Information Processing Systems
  34 (2021)

\bibitem{barto2013intrinsic}
Barto, A.G.: Intrinsic motivation and reinforcement learning. In: Intrinsically
  motivated learning in natural and artificial systems, pp. 17--47. Springer
  (2013)

\bibitem{ben2018game}
Ben-Porat, O., Tennenholtz, M.: A game-theoretic approach to recommendation
  systems with strategic content providers. In: Proceedings of the 32nd
  International Conference on Neural Information Processing Systems. pp.
  1118--1128 (2018)

\bibitem{benaim2009learning}
Bena{\"\i}m, M., Hofbauer, J., Hopkins, E.: Learning in games with unstable
  equilibria. Journal of Economic Theory  144(4),  1694--1709 (2009)

\bibitem{benoit1984finitely}
Benoit, J.P., Krishna, V., et~al.: Finitely repeated games  (1984)

\bibitem{birmpas2020optimally}
Birmpas, G., Gan, J., Hollender, A., Marmolejo, F., Rajgopal, N., Voudouris,
  A.: Optimally deceiving a learning leader in stackelberg games. Advances in
  Neural Information Processing Systems  (2020)

\bibitem{blackwell1956analog}
Blackwell, D.: An analog of the minimax theorem for vector payoffs. Pacific
  Journal of Mathematics  6(1),  1--8 (1956)

\bibitem{blackwell1954controlled}
Blackwell, D., et~al.: Controlled random walks. In: Proceedings of the
  international congress of mathematicians. vol.~3, pp. 336--338 (1954)

\bibitem{blum2007learning}
Blum, A., Monsour, Y.: Learning, regret minimization, and equilibria  (2007)

\bibitem{brafman2002efficient}
Brafman, R., Tennenholtz, M.: Efficient learning equilibrium. Advances in
  Neural Information Processing Systems  15 (2002)

\bibitem{braverman2018selling}
Braverman, M., Mao, J., Schneider, J., Weinberg, M.: Selling to a no-regret
  buyer. In: Proceedings of the 2018 ACM Conference on Economics and
  Computation. pp. 523--538 (2018)

\bibitem{brown1951iterative}
Brown, G.W.: Iterative solution of games by fictitious play. Activity analysis
  of production and allocation  13(1),  374--376 (1951)

\bibitem{cai2015optimum}
Cai, Y., Daskalakis, C., Papadimitriou, C.: Optimum statistical estimation with
  strategic data sources. In: Conference on Learning Theory. pp. 280--296. PMLR
  (2015)

\bibitem{calvo2006set}
Calv{\'o}-Armengol, A.: The set of correlated equilibria of 2x2 games. Working
  paper  (2006)

\bibitem{cesa2014regret}
Cesa-Bianchi, N., Gentile, C., Mansour, Y.: Regret minimization for reserve
  prices in second-price auctions. IEEE Transactions on Information Theory
  61(1),  549--564 (2014)

\bibitem{cesa2006prediction}
Cesa-Bianchi, N., Lugosi, G.: Prediction, learning, and games. Cambridge
  university press (2006)

\bibitem{chen2020hedging}
Chen, X., Peng, B.: Hedging in games: Faster convergence of external and swap
  regrets. Advances in Neural Information Processing Systems 33 pre-proceedings
  (NeurIPS 2020)  (2020)

\bibitem{chen2019learning}
Chen, Y., Liu, Y., Podimata, C.: Learning strategy-aware linear classifiers.
  Advances in Neural Information Processing Systems  33,  15265--15276 (2020)

\bibitem{cournot1838recherches}
Cournot, A.A.: Recherches sur les principes math{\'e}matiques de la th{\'e}orie
  des richesses. L. Hachette (1838)

\bibitem{daskalakis2021near}
Daskalakis, C., Fishelson, M., Golowich, N.: Near-optimal no-regret learning in
  general games. Advances in Neural Information Processing Systems  34 (2021)

\bibitem{daskalakis2018last}
Daskalakis, C., Panageas, I.: Last-iterate convergence: Zero-sum games and
  constrained min-max optimization. In: 10th Innovations in Theoretical
  Computer Science Conference (ITCS 2019). Schloss Dagstuhl-Leibniz-Zentrum
  fuer Informatik (2018)

\bibitem{daskalakis2016learning}
Daskalakis, C., Syrgkanis, V.: Learning in auctions: Regret is hard, envy is
  easy. In: 2016 ieee 57th annual symposium on foundations of computer science
  (focs). pp. 219--228. IEEE (2016)

\bibitem{deng2021autoBidding}
Deng, Y., Mao, J., Mirrokni, V., Zuo, S.: Towards efficient auctions in an
  auto-bidding world. In: Proceedings of the Web Conference 2021 (2021)

\bibitem{deng2019strategizing}
Deng, Y., Schneider, J., Sivan, B.: Strategizing against no-regret learners.
  Advances in Neural Information Processing Systems  32,  1579--1587 (2019)

\bibitem{dhillon2004plurality}
Dhillon, A., Lockwood, B.: When are plurality rule voting games
  dominance-solvable? Games and Economic Behavior  46(1),  55--75 (2004)

\bibitem{dong2018strategic}
Dong, J., Roth, A., Schutzman, Z., Waggoner, B., Wu, Z.S.: Strategic
  classification from revealed preferences. In: Proceedings of the 2018 ACM
  Conference on Economics and Computation. pp. 55--70 (2018)

\bibitem{eccles2019learning}
Eccles, T., Hughes, E., Kram{\'a}r, J., Wheelwright, S., Leibo, J.Z.: Learning
  reciprocity in complex sequential social dilemmas. arXiv preprint
  arXiv:1903.08082  (2019)

\bibitem{even2009convergence}
Even-Dar, E., Mansour, Y., Nadav, U.: On the convergence of regret minimization
  dynamics in concave games. In: Proceedings of the forty-first annual ACM
  symposium on Theory of computing. pp. 523--532 (2009)

\bibitem{foster1997calibrated}
Foster, D.P., Vohra, R.V.: Calibrated learning and correlated equilibrium.
  Games and Economic Behavior  21(1-2), ~40 (1997)

\bibitem{freund1999adaptive}
Freund, Y., Schapire, R.E.: Adaptive game playing using multiplicative weights.
  Games and Economic Behavior  29(1-2),  79--103 (1999)

\bibitem{fudenberg1999conditional}
Fudenberg, D., Levine, D.K.: Conditional universal consistency. Games and
  Economic Behavior  29(1-2),  104--130 (1999)

\bibitem{gan2019manipulating}
Gan, J., Guo, Q., Tran-Thanh, L., An, B., Wooldridge, M.: Manipulating a
  learning defender and ways to counteract. In: Proceedings of the 33rd
  International Conference on Neural Information Processing Systems. pp.
  8274--8283 (2019)

\bibitem{Jiarui2019imitativeFollower}
Gan, J., Xu, H., Guo, Q., Tran-Thanh, L., Rabinovich, Z., Wooldridge, M.:
  Imitative follower deception in stackelberg games. In: Proceedings of the
  2019 ACM Conference on Economics and Computation. p. 639–657. EC '19,
  Association for Computing Machinery (2019)

\bibitem{ghalme2021strategic}
Ghalme, G., Nair, V., Eilat, I., Talgam-Cohen, I., Rosenfeld, N.: Strategic
  classification in the dark  139,  3672--3681 (2021),
  \url{https://proceedings.mlr.press/v139/ghalme21a.html}

\bibitem{grigoriadis1995sublinear}
Grigoriadis, M.D., Khachiyan, L.G.: A sublinear-time randomized approximation
  algorithm for matrix games. Operations Research Letters  18(2),  53--58
  (1995)

\bibitem{haghtalab2016three}
Haghtalab, N., Fang, F., Nguyen, T.H., Sinha, A., Procaccia, A.D., Tambe, M.:
  Three strategies to success: Learning adversary models in security games. In:
  IJCAI (2016)

\bibitem{haghtalab2020maximizing}
Haghtalab, N., Immorlica, N., Lucier, B., Wang, J.Z.: Maximizing welfare with
  incentive-aware evaluation mechanisms. In: Proceedings of the Twenty-Ninth
  International Joint Conference on Artificial Intelligence, {IJCAI-20}. pp.
  160--166 (2020)

\bibitem{halac2020raising}
Halac, M., Kremer, I., Winter, E.: Raising capital from heterogeneous
  investors. American Economic Review  110(3),  889--921 (2020)

\bibitem{hannan1957lapproximation}
Hannan, J.: Approximation to bayes risk in repeated play. In: Contributions to
  the Theory of Games (AM-39), Volume III, pp. 97--139. Princeton University
  Press (1957)

\bibitem{hardt2016strategic}
Hardt, M., Megiddo, N., Papadimitriou, C., Wootters, M.: Strategic
  classification. In: Proceedings of the 2016 ACM conference on innovations in
  theoretical computer science. pp. 111--122 (2016)

\bibitem{hart2000simple}
Hart, S., Mas-Colell, A.: A simple adaptive procedure leading to correlated
  equilibrium. Econometrica  68(5),  1127--1150 (2000)

\bibitem{hart2013simple}
Hart, S., Mas-Colell, A.: Simple adaptive strategies: from regret-matching to
  uncoupled dynamics, vol.~4. World Scientific (2013)

\bibitem{jaques2019social}
Jaques, N., Lazaridou, A., Hughes, E., Gulcehre, C., Ortega, P., Strouse, D.,
  Leibo, J.Z., De~Freitas, N.: Social influence as intrinsic motivation for
  multi-agent deep reinforcement learning. In: International Conference on
  Machine Learning. pp. 3040--3049. PMLR (2019)

\bibitem{kalai2002geometric}
Kalai, A., Vempala, S.: Geometric algorithms for online optimization. In:
  Journal of Computer and System Sciences. Citeseer (2002)

\bibitem{kalai2005efficient}
Kalai, A., Vempala, S.: Efficient algorithms for online decision problems.
  Journal of Computer and System Sciences  71(3),  291--307 (2005)

\bibitem{kalai2010commitment}
Kalai, A.T., Kalai, E., Lehrer, E., Samet, D.: A commitment folk theorem. Games
  and Economic Behavior  69(1),  127--137 (2010)

\bibitem{kolumbus2021auctions}
Kolumbus, Y., Nisan, N.: Auctions between regret-minimizing agents. In:
  Proceedings of the ACM Web Conference 2022 (WWW ’22). pp. 100--111 (2022),
  \url{https://arxiv.org/pdf/2110.11855.pdf}

\bibitem{kulkarni2016hierarchical}
Kulkarni, T.D., Narasimhan, K., Saeedi, A., Tenenbaum, J.: Hierarchical deep
  reinforcement learning: Integrating temporal abstraction and intrinsic
  motivation. In: Advances in Neural Information Processing Systems. pp.
  3675--3683 (2016)

\bibitem{lavictoire2014program}
LaVictoire, P., Fallenstein, B., Yudkowsky, E., Barasz, M., Christiano, P.,
  Herreshoff, M.: Program equilibrium in the prisoner's dilemma via l{\"o}b's
  theorem. In: Workshops at the twenty-eighth AAAI conference on artificial
  intelligence (2014)

\bibitem{lei2020time}
Lei, K., Zhang, B., Li, Y., Yang, M., Shen, Y.: Time-driven feature-aware
  jointly deep reinforcement learning for financial signal representation and
  algorithmic trading. Expert Systems with Applications  140,  112872 (2020)

\bibitem{leibo2017multi}
Leibo, J.Z., Zambaldi, V., Lanctot, M., Marecki, J., Graepel, T.: Multi-agent
  reinforcement learning in sequential social dilemmas. In: Proceedings of the
  16th Conference on Autonomous Agents and MultiAgent Systems. pp. 464--473
  (2017)

\bibitem{levanon2021strategic}
Levanon, S., Rosenfeld, N.: Strategic classification made practical. In:
  Proceedings of the 38th International Conference on Machine Learning. vol.
  139, pp. 6243--6253. PMLR (2021)

\bibitem{li2019deep}
Li, Y., Zheng, W., Zheng, Z.: Deep robust reinforcement learning for practical
  algorithmic trading. IEEE Access  7,  108014--108022 (2019)

\bibitem{littman1994markov}
Littman, M.L.: Markov games as a framework for multi-agent reinforcement
  learning. In: Machine learning proceedings 1994, pp. 157--163. Elsevier
  (1994)

\bibitem{mailath2006repeated}
Mailath, G.J., Samuelson, L., et~al.: Repeated games and reputations: long-run
  relationships. Oxford university press (2006)

\bibitem{mansour2022strategizing}
Mansour, Y., Mohri, M., Schneider, J., Sivan, B.: Strategizing against learners
  in bayesian games. arXiv preprint arXiv:2205.08562  (2022)

\bibitem{mas1995microeconomic}
Mas-Colell, A., Whinston, M.D., Green, J.R., et~al.: Microeconomic theory,
  vol.~1. Oxford university press New York (1995)

\bibitem{mehta2022auction}
Mehta, A.: Auction design in an auto-bidding setting: Randomization improves
  efficiency beyond vcg. In: Proceedings of the ACM Web Conference 2022. pp.
  173--181 (2022)

\bibitem{mertikopoulos2018cycles}
Mertikopoulos, P., Papadimitriou, C., Piliouras, G.: Cycles in adversarial
  regularized learning. In: Proceedings of the Twenty-Ninth Annual ACM-SIAM
  Symposium on Discrete Algorithms. pp. 2703--2717. SIAM (2018)

\bibitem{monnot2017limits}
Monnot, B., Piliouras, G.: Limits and limitations of no-regret learning in
  games. The Knowledge Engineering Review  32 (2017)

\bibitem{moulin1979dominance}
Moulin, H.: Dominance solvable voting schemes. Econometrica: Journal of the
  Econometric Society pp. 1337--1351 (1979)

\bibitem{moulin1984dominance}
Moulin, H.: Dominance solvability and cournot stability. Mathematical social
  sciences  7(1),  83--102 (1984)

\bibitem{moulin1978strategically}
Moulin, H., Vial, J.P.: Strategically zero-sum games: the class of games whose
  completely mixed equilibria cannot be improved upon. International Journal of
  Game Theory  7(3-4),  201--221 (1978)

\bibitem{nachbar1990evolutionary}
Nachbar, J.H.: “evolutionary” selection dynamics in games: Convergence and
  limit properties. International journal of game theory  19(1),  59--89 (1990)

\bibitem{nekipelov2015econometrics}
Nekipelov, D., Syrgkanis, V., Tardos, E.: Econometrics for learning agents. In:
  Proceedings of the Sixteenth ACM Conference on Economics and Computation. pp.
  1--18 (2015)

\bibitem{nguyen2019imitative}
Nguyen, T., Xu, H.: Imitative attacker deception in stackelberg security games.
  In: IJCAI. pp. 528--534 (2019)

\bibitem{nguyen2019deception}
Nguyen, T.H., Wang, Y., Sinha, A., Wellman, M.P.: Deception in finitely
  repeated security games. In: Proceedings of the AAAI Conference on Artificial
  Intelligence. pp. 2133--2140 (2019)

\bibitem{nisan2007introduction}
Nisan, N., et~al.: Introduction to mechanism design (for computer scientists).
  Algorithmic game theory  9,  209--242 (2007)

\bibitem{noti2021bid}
Noti, G., Syrgkanis, V.: Bid prediction in repeated auctions with learning. In:
  Proceedings of the Web Conference 2021. pp. 3953--3964 (2021)

\bibitem{oesterheld2019robust}
Oesterheld, C.: Robust program equilibrium. Theory and Decision  86(1),
  143--159 (2019)

\bibitem{oesterheld2022safe}
Oesterheld, C., Conitzer, V.: Safe pareto improvements for delegated game
  playing. Autonomous Agents and Multi-Agent Systems  36(2),  1--47 (2022)

\bibitem{osborne1994course}
Osborne, M.J., Rubinstein, A.: A course in game theory. MIT press (1994)

\bibitem{papadimitriou2018nash}
Papadimitriou, C., Piliouras, G.: From nash equilibria to chain recurrent sets:
  An algorithmic solution concept for game theory. Entropy  20(10),  782 (2018)

\bibitem{piliouras2021optimal}
Piliouras, G., Sim, R., Skoulakis, S.: Optimal no-regret learning in general
  games: Bounded regret with unbounded step-sizes via clairvoyant mwu. arXiv
  preprint arXiv:2111.14737  (2021)

\bibitem{rahwan2019machine}
Rahwan, I., Cebrian, M., Obradovich, N., Bongard, J., Bonnefon, J.F., Breazeal,
  C., Crandall, J.W., Christakis, N.A., Couzin, I.D., Jackson, M.O., et~al.:
  Machine behaviour. Nature  568(7753),  477--486 (2019)

\bibitem{robinson1951iterative}
Robinson, J.: An iterative method of solving a game. Annals of mathematics pp.
  296--301 (1951)

\bibitem{segal1999contracting}
Segal, I.: Contracting with externalities. The Quarterly Journal of Economics
  114(2),  337--388 (1999)

\bibitem{segal2007tit}
Segal, U., Sobel, J.: Tit for tat: Foundations of preferences for reciprocity
  in strategic settings. Journal of Economic Theory  136(1),  197--216 (2007)

\bibitem{shalev2011online}
Shalev-Shwartz, S., et~al.: Online learning and online convex optimization.
  Foundations and trends in Machine Learning  4,  107--194 (2011)

\bibitem{shapley1964some}
Shapley, L.: Some topics in two-person games. Advances in game theory  52,
  1--29 (1964)

\bibitem{shi2020learning}
Shi, Z.R., Procaccia, A.D., Chan, K.S., Venkatesan, S., Ben-Asher, N., Leslie,
  N.O., Kamhoua, C., Fang, F.: Learning and planning in the feature deception
  problem. In: International Conference on Decision and Game Theory for
  Security. pp. 23--44. Springer (2020)

\bibitem{csimcsek2006intrinsic}
{\c{S}}im{\c{s}}ek, {\"O}., Barto, A.G.: An intrinsic reward mechanism for
  efficient exploration. In: Proceedings of the 23rd international conference
  on Machine learning. pp. 833--840 (2006)

\bibitem{singh2004intrinsically}
Singh, S., Barto, A.G., Chentanez, N.: Intrinsically motivated reinforcement
  learning. In: Proceedings of the 17th International Conference on Neural
  Information Processing Systems (2004)

\bibitem{singh2010intrinsically}
Singh, S., Lewis, R.L., Barto, A.G., Sorg, J.: Intrinsically motivated
  reinforcement learning: An evolutionary perspective. IEEE Transactions on
  Autonomous Mental Development  2(2),  70--82 (2010)

\bibitem{syrgkanis2015fast}
Syrgkanis, V., Agarwal, A., Luo, H., Schapire, R.E.: Fast convergence of
  regularized learning in games. Advances in Neural Information Processing
  Systems  28,  2989--2997 (2015)

\bibitem{tennenholtz2004program}
Tennenholtz, M.: Program equilibrium. Games and Economic Behavior  49(2),
  363--373 (2004)

\bibitem{tennenholtz2019rethinking}
Tennenholtz, M., Kurland, O.: Rethinking search engines and recommendation
  systems: a game theoretic perspective. Communications of the ACM  62(12),
  66--75 (2019)

\bibitem{algotrade2013}
Treleaven, P., Galas, M., Lalchand, V.: Algorithmic trading review. Commun. ACM
   56(11),  76–85 (nov 2013), \url{https://doi.org/10.1145/2500117}

\bibitem{young2004strategic}
Young, H.P.: Strategic learning and its limits. OUP Oxford (2004)

\bibitem{ZHOU2019186}
Zhou, H., Kalev, P.S.: Algorithmic and high frequency trading in asia-pacific,
  now and the future. Pacific-Basin Finance Journal  53,  186--207 (2019)

\bibitem{zinkevich2003online}
Zinkevich, M.: Online convex programming and generalized infinitesimal gradient
  ascent. In: Proceedings of the 20th international conference on machine
  learning. pp. 928--936 (2003)

\end{thebibliography}


\section*{Checklist}


\begin{enumerate}

\item For all authors...
\begin{enumerate}
  \item Do the main claims made in the abstract and introduction accurately reflect the paper's contributions and scope?
    \answerYes{}
  \item Did you describe the limitations of your work?
    \answerYes{}
  \item Did you discuss any potential negative societal impacts of your work?
    \answerYes{As in every game theoretic study, outcomes that may be useful to one player may reduce the utility to another, depending on the nature of the underlying game. We provide fundamental theoretical analysis for understanding the incentives that users of learning algorithms may have to manipulate their own agents and the potential outcomes of such manipulations.}  
  \item Have you read the ethics review guidelines and ensured that your paper conforms to them?
    \answerYes{}
\end{enumerate}

\item If you are including theoretical results...
\begin{enumerate}
  \item Did you state the full set of assumptions of all theoretical results?
    \answerYes{}
        \item Did you include complete proofs of all theoretical results?
    \answerYes{Due to the space limitation, all proofs are deferred to the supplementary materials.}
\end{enumerate}

\item If you ran experiments...
\begin{enumerate}
  \item Did you include the code, data, and instructions needed to reproduce the main experimental results (either in the supplemental material or as a URL)?
    \answerYes{The code and requirements for our simulations are provided as part of the supplementary materials.}
  \item Did you specify all the training details (e.g., data splits, hyperparameters, how they were chosen)?
    \answerNA{}
        \item Did you report error bars (e.g., with respect to the random seed after running experiments multiple times)?
    \answerNA{The number $T$ of game repetitions in our simulations is purposely chosen such that the variance of the time average of player payoffs between
 independent simulations is small and error bars were not informative.}
        \item Did you include the total amount of compute and the type of resources used (e.g., type of GPUs, internal cluster, or cloud provider)?
    \answerNA{Our simulations do not require any special resources and can be run on a CPU in a standard PC.} 
\end{enumerate}

\item If you are using existing assets (e.g., code, data, models) or curating/releasing new assets...
\begin{enumerate}
  \item If your work uses existing assets, did you cite the creators?
    \answerNA{}
  \item Did you mention the license of the assets?
    \answerNA{}
  \item Did you include any new assets either in the supplemental material or as a URL?
    \answerNA{}
  \item Did you discuss whether and how consent was obtained from people whose data you're using/curating?
    \answerNA{}
  \item Did you discuss whether the data you are using/curating contains personally identifiable information or offensive content?
    \answerNA{}
\end{enumerate}

\item If you used crowdsourcing or conducted research with human subjects...
\begin{enumerate}
  \item Did you include the full text of instructions given to participants and screenshots, if applicable?
    \answerNA{}
  \item Did you describe any potential participant risks, with links to Institutional Review Board (IRB) approvals, if applicable?
    \answerNA{}
  \item Did you include the estimated hourly wage paid to participants and the total amount spent on participant compensation?
    \answerNA{}
\end{enumerate}

\end{enumerate}

\clearpage

\appendix

\section{Convergence of Regret-Minimizing Agents}\label{sec:convergence}

We consider repeated-game settings, in which the same finite group of automated agents repeatedly play a fixed game with bounded utilities on behalf of their users, with one agent per user. We focus on learning agents that are implemented as regret-minimization algorithms. The (external) regret of player $i$ at time $T$, given a history of play $(\textbf{a}^1, ... ,\textbf{a}^T)$, is defined as the difference between the optimal utility from using a fixed action in hindsight and the actual utility: $R_i^T = \max_a \sum_{t=1}^T u_i(a, \textbf{a}_{-i}^t) - u_i(a_i^t, \textbf{a}_{-i}^t)$, where $a_i^t$ is the action of player $i$ at time $t$ and $\textbf{a}_{-i}^t$ denotes the action profile of the other players at time $t$. As usual, regret-minimization algorithms are stochastic, and whenever we talk about the limit behavior, we consider a sequence of algorithms with $T \rightarrow \infty$ and with probability approaching $1$. An agent $i$ is said to be ``regret minimizing'' if $R_i^T/T \rightarrow 0$ almost surely as $T \rightarrow \infty$. 
A joint distribution over the players' actions is said to be a coarse correlated equilibrium if under this distribution all players have on expectation at most zero regret. 

We use the following notation to describe the empirical distributions of the agents' dynamics. We denote by $\Delta$ the space of probability distributions over action profiles, by $\textbf{p}^T_t\in\Delta$ the empirical distribution of actions after $t$ rounds in a sequence of $T\geq t$ repetitions of the game, and by   
$p_t^T(\textbf{a})$ the empirical frequency of an action profile $\textbf{a}$ at the end of round $t$ of the repeated game. 

It is well known that the regret-minimization property ensures that the inequalities that define the CCE condition are all satisfied for the empirical time-average utilities over $T$ steps to within a diminishing error term.  This implies that, in $\Delta$ (the space of probability distributions over the agents' joint actions), the empirical joint-action distribution, $\textbf{p}^T_t$, must get arbitrarily close to the \emph{polytope of CCE distributions}.\footnote{This follows directly from the compactness of the space of distributions, assuming that utilities are bounded.} One may hope that the empirical distribution converges to some specific CCE, $p^*$, and even be encouraged by the many known examples (starting with the matching pennies game) in which, even though the mixed strategies of the players do not converge, the time average of the empirical-play dynamics does converge to a specific CCE \cite{bailey2021stochastic,bailey2018multiplicative,benaim2009learning}.

However, not only is this hope not always justified \cite{shapley1964some}, but our following observation shows that whenever there is more than a single CCE no convergence is guaranteed. Not only may regret-minimization dynamics not converge at all, but even the time averages of the action distribution and utilities may keep changing over time. 

\begin{proposition}\label{thm:convergence}
For every finite game in which the set of CCEs is not a singleton and for every pair of distinct CCE distributions in that set,    
there exist regret-minimizing algorithms for the players whose empirical time-average joint dynamics do not 
converge at all and oscillate between getting arbitrarily close to each of these two CCEs.
\end{proposition}

Thus, when allowing for \emph{general} regret-minimizing agents, the only types of games in which the question of convergence is resolved are games in which the CCE is unique; in these games, regret-minimization dynamics must converge to the single CCE. For games with multiple CCEs, it is required to have additional information on the types of regret-minimization algorithms that are applied, that will allow to analyze the convergence properties of their dynamics, as we do in the companion paper \cite{kolumbus2021auctions} for a large class of natural regret-minimization algorithms in first-price and second-price auctions. 

Before proceeding to the proof of Proposition \ref{thm:convergence}, we need to formally define what notions of convergence we are looking at. 
We propose the following definitions of convergence as general and concrete notions that are compatible with the standard models of regret-minimization algorithms that focus on algorithms $ALG^T$, each of which is targeted to a fixed time horizon $T$, and looks at a sequence of such algorithms as\footnote{
An alternative formalism would consider a single algorithm with an infinite horizon, as in \cite{bailey2021stochastic,bailey2018multiplicative}, and look at its intermediate results at different times $T$.  We prefer following the fixed-horizon formalism as it is the most commonly used one and since in the infinite-horizon definitions one must have an appropriately decreasing ``update parameter.''} $T \rightarrow \infty$. 
All the following notions of convergence concern the average-iterate (i.e., the empirical distributions of the agents’ dynamics). Note that average-iterate convergence does not imply last-iterate convergence (whereas the converse is true), and so under all the following definitions, a CCE may include dynamical patterns such as cycles or recurrent sets \cite{mertikopoulos2018cycles,papadimitriou2018nash}.

The first notion of convergence that we consider is the one closest to the definition of the regret-minimization property. As mentioned above, the dynamics of regret-minimizing agents approach the set of coarse correlated equilibria in the space of utilities. The following definition deals with the space of distributions over joint action profiles.  

\begin{definition} The dynamics \emph{approach a set} $S\subseteq \Delta$ of distributions if for every $\epsilon>0$ there exists $T_0(\epsilon)$ such that for every $T>T_0$ with probability at least $1-\epsilon$ it holds that $\underset{\textbf{p} \in S}{\inf} |\textbf{p}_T^T - \textbf{p}| < \epsilon$.
\end{definition}

The next definition describes a different property of the dynamics. Basically, this property means that after a sufficiently long time, the distribution of the empirical play stabilizes and does not change much. Notice that this property allows for having different outcomes for different values of $T$, or even for different instances of the dynamics with the same algorithms and the same $T$.  
\begin{definition} The dynamics are \emph{self-convergent} if for every $\epsilon>0$ there exists $T_0(\epsilon)$ such that for every $T>T_0$ with probability at least $1-\epsilon$ it holds that for every $\epsilon T < t \leq T$, $|\textbf{p}_t^T - \textbf{p}_T^T| < \epsilon$.
\end{definition}

Finally, the following definition describes convergence to a single distribution. This definition will be useful in our analysis of meta-games of games with a single CCE, in which the dynamics must converge to that CCE for any set of regret-minimization algorithms.\footnote{This definition would correspond to
convergence to $\textbf{p}$ almost surely in the infinite-horizon model.}
 
\begin{definition}\label{def:converge-to-p} 
The dynamics \emph{converge to a distribution} $\textbf{p}\in \Delta $ if for every $\epsilon>0$ there exists $T_0(\epsilon)$ such that for every $T>T_0$ with probability at least $1-\epsilon$ it holds that for every $\epsilon T < t \leq T$, $|\textbf{p}_t^T - \textbf{p}| < \epsilon$.
\end{definition}

It is not difficult to see that the above definition of convergence to a distribution is equivalent to the combination of the first two definitions shown above, as follows. 
\begin{observation} 
The dynamics converge to the distribution $\textbf{p}$ if and only if the dynamics are 
self-convergent and approach the unit set $\{\textbf{p}\}\subset \Delta$.
\end{observation}

Proposition \ref{thm:convergence}, which considers Definition \ref{def:converge-to-p} of convergence, shows that in any game with more than a single CCE  
there exist regret-minimizing algorithms for the players whose empirical time-average joint dynamics do not converge at all. 
The proof takes as a starting point the fact that for any CCE in any game there exist regret-minimization dynamics that converge to it. 
To establish convergence to a specific CCE, as in \cite{monnot2017limits}, we look at dynamics in which all agents play according to a schedule that yields that CCE as its time average, and in any case of deviation by any subset of the other players, the algorithms divert to playing a standard regret-minimization algorithm in the remaining time. The idea of the proof of Proposition \ref{thm:convergence} is that instead of using an action schedule that converges to a single CCE, we construct dynamics that alternate between two such schedules, and show that if this alternation slows down at a sufficient rate, the time average oscillates between arbitrarily approaching each of the two pre-determined CCEs, while the regret-minimization property of each agent is preserved. 

\begin{proof} (Proposition \ref{thm:convergence}):
We start with the following claim: \emph{for every finite game and every CCE distribution of that game there exist regret-minimizing algorithms for the players whose joint dynamics converge to the given CCE}. This result was previously shown in  \cite{monnot2017limits} in an infinite-horizon setting.  
We technically re-prove this claim here to make the proof compatible with the finite-horizon setting and, specifically, with Definition \ref{def:converge-to-p} of convergence given above. 

Let $\textbf{p}$ be a CCE distribution of a finite $n$-player game, and let $M=|A_1 \times...\times A_n|$ denote the size of the joint action space of the players, where $A_i$ is the action space of agent $i$. We assume for simplicity that every probability in $\textbf{p}$ is a rational number. Let $\textbf{a}^k$, $k\in[M]$, be the action profile that has the $k$'-th highest probability in $\textbf{p}$ (ties are broken in favor of the action tuple with the higher index in $\textbf{p}$) and denote by $a_i^k$ the action of player $i$ in the action tuple $\textbf{a}^k$. Let $T_k = 1/\Pr({\textbf{a}^k})$ if $\Pr({\textbf{a}^k}) > 0$ or $T_k = 0$ otherwise and let $T_0 = 0$. Define the mapping $k(t) = k \text{ s.t. } \sum_{s=0}^{k-1} T_s < mod(t, \sum_{s=1}^M T_s) \leq \sum_{s=0}^{k} T_s$. 

Consider the following ``$\textbf{p}$-schedule'' algorithm for agent $i$: at every time $t = 1,...,T$, the algorithm plays action $a^{k(t)}_i$. 
After every action, the algorithm observes the actions of the other agents. If at least one of the other agents deviated from its schedule, i.e., if there exists a player $j$ that played an action different from $a^{k(t)}_j$, then the algorithm stops playing according to $a^{k(t)}_i$ and switches to playing the ``unconditional regret-matching'' algorithm \cite{hart2000simple,hart2013simple} in the remaining time. 

This algorithm is clearly regret-minimizing: on the one hand, if any agent deviates, then all agents play the unconditional regret-matching algorithm, which is regret-minimizing.\footnote{This algorithm is especially simple here since it is does not require any parameter to be fitted to the remaining time after the step at which a deviation occurred. Other algorithms can also be used with proper adjustments.} On the other hand, if all agents are playing according to this algorithm then in every period of length $\tau = \sum_{s=1}^M T_s$ the empirical action distribution exactly equals the CCE distribution $\textbf{p}$. Since all the utilities are assumed to be bounded, this implies that the regret is vanishing in the limit $T\rightarrow \infty$ (since the regret accumulated in any ``partial cycle'' of the schedule, of length less than $\tau$, is bounded by a constant and thus vanishes in the time average). 

It also follows from a similar argument that in the dynamics in which all agents play according to the algorithm described above, the empirical distribution of play converges to $\textbf{p}$ (we will call such dynamics in which all agents play the $\textbf{p}$-schedule algorithm ``$\textbf{p}$-schedule dynamics''). Formally, at the end of every period $\tau$ the empirical action distribution equals $\textbf{p}$, and thus the empirical distribution at time $t > \tau$ can be written as $\textbf{p}_t^T = \frac{1}{t}\left( \tau \cdot {\left\lfloor t/\tau \right\rfloor} \cdot \textbf{p} + mod(t,\tau)\cdot \textbf{x} \right)$, where $\textbf{x}$ can be any arbitrary distribution obtained in a partial cycle of the schedule, or it can be equal to $\textbf{p}$ if $t$ completes a full cycle. To show that Definition \ref{def:converge-to-p} holds, let $\epsilon > 0$ such that $ T_0(\epsilon) > 2 \tau / \epsilon$. To bound the distance of the empirical action distribution from $\textbf{p}$, consider the vector $\textbf{x}=-\textbf{p}$ with weight $\tau$. Then, for every $T > T_0$ with probability $1$ (as there in no stochasticity in these dynamics) for every $t$ s.t. $\epsilon T < t \leq T$  it holds that $|\textbf{p}_t^T - \textbf{p}| \leq | \frac{1}{t}\big((t-\tau)\textbf{p} - \tau \textbf{p}\big) - \textbf{p}| \leq \frac{2\tau}{t} < \epsilon$. Thus, we have that every CCE has regret-minimization dynamics that converge to it.

Next, consider any finite game that has more than a single CCE and let $\textbf{p}_1$ and $\textbf{p}_2$ be two distinct distributions in the set of its CCE distributions. Notice that since the set of CCEs in every game is a convex set, also every weighted average of $\textbf{p}_1$ and $\textbf{p}_2$ is a CCE itself. Consider the dynamics of agents that all switch between $\textbf{p}_1$-schedule dynamics with cycle time $\tau_1$ and $\textbf{p}_2$-schedule dynamics with cycle time $\tau_2$. Let $\epsilon > 0$ such that $\epsilon < 1/\max(\tau_1, \tau_2)$ and let $\alpha = \left\lceil 1/ \epsilon^2 \right\rceil$. We will define the series $\Gamma_c = (2 \alpha)^c,$ $c=1,2,...$ to be the number of full cycles in which each dynamic is played before switching to the other dynamic. That is, the lengths of the periods in which each dynamic is played are $\tau_1\Gamma_c$ if $c$ is odd or $\tau_2\Gamma _c$ if $c$ is even, i.e., $(2 \alpha\tau_1, 4 \alpha^2 \tau_2, 8 \alpha^3 \tau_1, 16 \alpha^4 \tau_2,...)$. By the above proof for the convergence of $\textbf{p}$-schedule dynamics we have that by the end of the first period of length $2 \alpha\tau_1$ the empirical distribution of play reaches a distance of less than $\epsilon$ from the CCE distribution $\textbf{p}_1$. Next, by the end the second period, i.e., at time $t = 2 \alpha\tau_1 + 4 \alpha^2 \tau_2$, the distance of the empirical distribution from the CCE distribution $\textbf{p}_2$ is $|\textbf{p}_t^T - \textbf{p}_2| \leq \left|\frac{1}{4 \alpha^2 \tau_2 + 2 \alpha\tau_1}\Big( (4 \alpha^2 \tau_2)\textbf{p}_2 + (2 \alpha\tau_1)\textbf{p}_1 \Big) - \textbf{p}_2\right| \leq \left|\frac{1}{4 \alpha^2 \tau_2 + 2 \alpha\tau_1}\Big( (4 \alpha^2 \tau_2)\textbf{p}_2 - (2 \alpha\tau_1)\textbf{p}_2 \Big) - \textbf{p}_2\right| = \left|\frac{4 \alpha^2 \tau_2 - 2 \alpha\tau_1}{4 \alpha^2 \tau_2 + 2 \alpha\tau_1} - 1\right| = \frac{4 \alpha\tau_1}{4 \alpha^2 \tau_2 + 2 \alpha\tau_1} < \epsilon$. 

The same argument holds at the end of the subsequent periods as well. Thus, given a sufficiently long time $T$, the average empirical distribution $\textbf{p}_t^T $ oscillates between getting arbitrarily close to each of the two CCE distributions $\textbf{p}_1$ and $\textbf{p}_2$, with an oscillation period that slows down exponentially, and hence there is no convergence of the average empirical play. 
\end{proof}

\section{Dominance-Solvable Games}\label{sec:appendix-DS-games}

The following lemma is a formalization of a well-known result that dominance-solvable games have a single CCE. The direct implication is that these games are stable in the dynamical sense, as the empirical time average of regret-minimization dynamics must converge to the unique Nash equilibrium outcome. 
For completeness, we provide here the formal statement and a simple proof. Some relevant references in this context include \cite{calvo2006set,moulin1984dominance,nachbar1990evolutionary}.

\begin{lemma}\label{thm:unique-CCE-in-DS-games}
In any dominance-solvable game the only CCE is the unique pure Nash equilibrium. 
\end{lemma}

\begin{proof} 
Consider any dominance-solvable game. Denote the Nash equilibrium joint action distribution by $\sigma_{NE}$. In this distribution there is probability one for the pure Nash equilibrium joint action profile and probability zero for all other action profiles. First, note that dominance-solvable games have a unique Nash equilibrium, and that every Nash equilibrium is also a CCE, and so $\sigma_{NE}$ is a CCE. Next, assume by way of contradiction that there is another distribution $\sigma \neq \sigma_{NE}$ that is also a CCE. Fix any order $\textbf{s}$ of iterated elimination of strictly dominated strategies in the game, such that $s_i$ is the $i$'th eliminated action. Since $\sigma$ is different from the pure Nash equilibrium distribution, there exists an action that has the smallest index $i$ such that $s_i$ is played with positive probability according to $\sigma$, but played with zero probability according to $\sigma_{NE}$. If $i=1$ then $s_i$ is a strictly dominated strategy. Therefore, the player that plays action $s_1$ with a finite frequency (with high probability) must accumulate linear regret (since action $i=1$ is specifically dominated also by the best fixed strategy in hindsight which yields, by definition, zero regret). Since the CCE condition is equivalent to the requirement that all players have, with high probability, a sub-linear regret over time, $\sigma$ cannot be a CCE, a contradiction. Hence $i>1$. However, if $i=2$, since action $1$ is played with zero probability, a player that plays $s_2$ with a finite frequency must also accumulate linear regret, which again violated the CCE condition, and thus $i>2$. The same argument holds until reaching the actions in $\textbf{s}$ which comprise together the pure Nash equilibrium profile (a single action for each player). Thus, we reach the contradiction $\sigma = \sigma_{NE}$ and so the unique pure Nash equilibrium of a dominance-solvable game is also its unique CCE.  
\end{proof}

\begin{proof} (Theorem 1):   
Consider any dominance-solvable $n\times m$ game in which there is a player who's Stackelberg value (i.e., his utility in a pure-strategy Stackelberg equilibrium of the game where he plays the first action) is higher than his utility in the unique Nash equilibrium of the game. Assume for convenience and without loss of generality that this is player $1$. For a sufficiently unrestricted parameter space, this player can declare his Stackelberg strategy as a dominant strategy to his agent. In this case, regret-minimization dynamics will quickly reach distributions of play where only this dominating strategy is played, regardless of the actions of the other agent. Specifically, if the opponent provides a truthful declaration to his agent, the declared game has a unique Nash equilibrium which is exactly the Stackelberg outcome. Thus a unilateral manipulation of player $1$ can give him his Stackelberg value and so improve his utility. 
\end{proof}

\section{Cournot Competition Games}\label{sec:appendix-Cournot}

As described in the main text, we consider linear Cournot competition games \cite{cournot1838recherches,even2009convergence,mas1995microeconomic}, where player $1$ produces quantity $q_1 \in \mathcal{R}^+$ with a per-unit production cost of $c_1$ (such that his total production cost is $c_1\cdot q_1$) and player $2$ produces quantity $q_2 \in \mathcal{R}^+$ with a per-unit production cost of $c_2$. The utilities of the players are
{\small
$
u_1 = q_1(a - b(q_1 + q_2) - c_1)
$, 
and 
$  
u_2 = q_2(a - b(q_1 + q_2) - c_2), 
$
}
where $a$ and $b$ are commonly known positive constants. The Nash equilibrium of the game depends on the parameters as follows. 
If $a + c_2 - 2c_1 > 0$ and $a + c_1 - 2c_2>0$, then the Nash equilibrium is  
{\small
$q_1^{true} = \frac{1}{3b}(a + c_2 - 2c_1)$  
and   
$q_2^{true} = \frac{1}{3b}(a + c_1 - 2c_2)$. 
}
If $a + c_1 - 2c_2>0$ and $c_1 < a$, the equilibrium is 
{\small
$q_1^{true} = \frac{a-c_1}{2b}$  
and
$q_2^{true} = 0$,   
}
and symmetrically, if $a + c_2 - 2c_1>0$ and $c_2 < a$, the equilibrium is $q_1 = 0$ and $q_2^{true} = \frac{a-c_2}{2b}$. 
Otherwise, in the Nash equilibrium both players produce zero.

Thus, there are four parameter regions of interest according to the four possible types of unique Nash equilibria of the game, as illustrated in Figure \ref{fig:cournot_regions}. The parameter region $A=\{c_1,c_2|a + c_2 - 2c_1 > 0, \ a + c_1 - 2c_2 > 0, \ c_1>0, \ c_2>0\}$ shown in the figure is the region where both agents produce positive quantities (the shaded areas in region $A$ in the figure relate to equilibria of the meta-game; see below). The parameter regions $B=\{c_1,c_2|a + c_1 - 2c_2 > 0, \  0 < c_1 < a, \ c_2>0\}$ and $C=\{c_1,c_2|a + c_2 - 2c_1 > 0, \  0 < c_2 < a, \ c_1>0\}$ are regions where only one player produces a positive quantity. In the remaining region, region $D=\{c_1,c_2|c_1, c_2 \geq a\}$, both agents produce zero. 

Consider the meta-game defined for this game where the true types of the players are their per-unit production costs $c_1,c_2$. Player $1$ declares to his agent a value\footnote{Notice that any declaration $x > a$ leads to zero production, and so it is equivalent to declaring $x=a$. Thus, $x_i \leq a$ suffices for a full description of the meta-game, where region $D$ is described by the point $x_1=x_2=a$.} $x_1 \in \mathcal{R}^+$ and player $2$ declares to his agent a value $x_2 \in \mathcal{R}^+$. The agents then interact repeatedly. Since this game is known to be ``socially concave'' \cite{even2009convergence}, the time-average dynamics of regret-minimizing agents must converge to the Nash equilibrium of the game that is defined by the parameters $x_1,x_2$ provided by the users. 

Given the declarations $x_1, x_2$, we can identify the relevant parameter region in which the declared game lies, substitute the appropriate equilibrium production levels of the agents in the utility functions for the players, and obtain the utilities of the players. The utilities in the four regions as functions of the declarations $x_1,x_2$, are 
{\small
$
u_1 = \frac{1}{9b}(a + x_2 - 2x_1)(a + x_1 + x_2 - 3c_1)$,  
$
u_2 = \frac{1}{9b}(a + x_1 - 2x_2)(a + x_1 + x_2 - 3c_2)$ 
} 
in region $A$,  
$u_1 = \frac{1}{4b}(a - x_1)(a + x_1 - 2c_1)$,    
$u_2 = 0$ 
in region $B$, and 
$u_1 = 0$, 
$u_2 = \frac{1}{4b}(a - x_2)(a + x_2 - 2c_2)$ 
in region $C$. In region $D$ there is no production and thus zero utility.

The following lemma specifies the equilibria of the meta-game for those cases where in these equilibria both agents produce positive quantities (i.e., where the equilibrium declarations are in region $A$). 

\begin{lemma}\label{thm:meta-game-both-produce-NE} 
Let $x_1^* = \frac{1}{5}(8c_1 - 2c_2 - a)$ and $x_2^* = \frac{1}{5}(8c_2 - 2c_1 - a)$.
If in the equilibrium of the meta-game both players produce positive quantities, then the equilibrium declarations are $x_1^*, x_2^*$ if $x_1^*, x_2^* \geq 0$; $(x_1 = \frac{1}{4}(6c_1-a)^+$, $x_2 = 0)$ if $x_2^* < 0, x_1^*\geq 0$;  $(x_1 = 0$, $x_2 = \frac{1}{4}(6c_2-a)^+)$ if $x_1^* < 0, x_2^* \geq 0$; and $(x_1=0,x_2=0)$ otherwise.
\end{lemma}

\begin{proof} 
If in the equilibrium of the meta-game both players produce positive quantities, then the equilibrium can be found using the derivatives of the utilities as follows.  
\small
\begin{equation*}
\frac{\partial u_1}{\partial x_1} = \frac{1}{9b}(6c_1-a-4x_1-x_2) = 0, \ \quad \
\frac{\partial u_2}{\partial x_2} = \frac{1}{9b}(6c_2-a-x_1-4x_2) = 0.
\end{equation*}
\normalsize
$x_1^* = \frac{1}{5}(8c_1 - 2c_2 - a)$ and $x_2^* = \frac{1}{5}(8c_2 - 2c_1 - a)$ are the unique solution to these equations. If $x_1, x_2$ are positive then they are the Nash equilibrium of the meta-game. If $x_1 = \frac{1}{5}(8c_1 - 2c_2 - a) < 0$ then the utility of player $1$ is decreasing in his declaration, and thus declaring $x_1 = 0$ is his best-reply to any declaration of the other player. The best-reply of player $2$ is then obtained from the above derivatives by substituting $x_1=0$ in the second equation, yielding $x_2 = \frac{1}{4}(6c_2 - a)$. If this expression is non-negative then it is the best-reply of player $2$. If this expression is negative then the utility of player $2$ is decreasing in $x_2$ and his best-reply is $x_2=0$. The same argument holds for player $2$; if $x_2 = \frac{1}{5}(8c_2 - 2c_1 - a)  < 0$ then declaring $x_2=0$ is the best-reply of player $2$ to any declaration of player $1$ and the best-reply of player $1$ is then $x_1 = \frac{1}{4}(6c_1 - a)^+$. 
\end{proof}

Theorem \ref{thm:cournot-meta-game-NE} characterizes the types of equilibria of the meta-game when in the equilibrium of the game with the true parameters both players produce positive quantities. To prove the theorem, we consider its following technical restatement. 

\begin{theorem*} (Restatement of Theorem \ref{thm:cournot-meta-game-NE}):  
Consider a two-player linear Cournot competition with positive linear costs where both players produce positive quantities in the Nash equilibrium. 
\begin{enumerate}[leftmargin=*]

\item If the production costs $c_1, c_2$ of the two players are sufficiently low such that $c_1,c_2 < a/2$ or sufficiently close such that $\frac{1}{2}(3c_2-a)    < c_1 <\frac{1}{3}(2c_2+a)$, then in the equilibrium of the meta-game both players declare to their agents values that are lower than their true production costs, produce larger quantities than those they produce in the Nash equilibrium of the game with the truthful reports (and thus the price is lower), and have lower utilities.

\item If the production cost of one of the players is at least $a/2$ and the cost of the other player is sufficiently low, namely, $c_2 \geq a/2$ and $c_2 \geq \frac{1}{3}(2c_1+a)$, or $c_1 \geq a/2$ and $c_1 \geq \frac{1}{3}(2c_2+a)$, then in the equilibrium of the meta-game the player with the low production cost declares a value that is lower than his true cost and produces alone. The quantity produced by this player alone is larger than the total quantity produced by both players in the Nash equilibrium of the game with the truthful reports (and thus the price is lower), and the utility for this player is higher than his utility in the Nash equilibrium of the game with the truthful reports. 
\end{enumerate}
\end{theorem*}

\begin{proof}
We start with the following lemma concerning the best-replies of the players. 
\begin{lemma} 
If $0< c_1 < a/2$, the best-reply of player $1$ to any declaration $x_2 < a$ of player $2$ is strictly less than player $1$'s true cost $c_1$. 
\end{lemma}
\begin{proof}
Assume that player $2$ declares a cost $0 \leq x_2 < a$. The utility of player $1$ is $u_1 = \frac{1}{9b}(a + x_2 - 2x_1)(a + x_1 + x_2 - 3c_1)$. The best-reply, $\argmax_{x_1}(u_1)$, is obtained from $\frac{\partial u_1}{\partial x_1} = \frac{1}{9b}(6c_1-a-4x_1-x2) = 0$, resulting in $x_1 = \frac{1}{4}(6c_1 - a - x_2)$.
If this expression is non-negative, then this is the best-reply; if it is negative, then the utility of player $1$ is decreasing in $x_1$ and the best-reply is $x_1=0$. If the best-reply is zero then it is less than $c_1$ as required. If the best-reply is positive then, since $c_1 < a/2$, it holds that $x_1 = \frac{1}{4}(6c_1 - a - x_2) \leq \frac{1}{4}(6c_1 - a) < \frac{1}{4}(6c_1 - 2c_1) = c_1$. 
\end{proof}
It thus follows that if $c_1, c_2 < a/2$, specifically, also in the Nash equilibrium declaration profile, which is a mutual best-reply profile, both declarations are less than the true costs.

Next, consider the case where the production costs of the two players, $c_1, c_2$, are at least $a/2$ and it holds that $\frac{1}{2}(3c_2-a) < c_1 <\frac{1}{3}(2c_2+a)$. The equilibrium of the meta-game is then $x_1 = \frac{1}{5}(8c_1 - 2c_2 - a)$ and $x_2 = \frac{1}{5}(8c_2 - 2c_1 - a)$ (by Lemma \ref{thm:meta-game-both-produce-NE}). Since $c_2 > \frac{1}{2}(3c_1 - a)$, it holds that $x_1 = \frac{1}{5}(8c_1 - 2c_2 - a) < \frac{1}{5}(8c_1 -3c_1+a-a)=c_1$, and similarly, since $c_1 > \frac{1}{2}(3c_2 - a)$, it holds that $x_2 = \frac{1}{5}(8c_2 - 2c_1 - a) < \frac{1}{5}(8c_2 -3c_2+a-a)=c_2$. Thus, the decelerations of both players are lower than their true costs, and each player produces a larger quantity than in the truthful equilibrium. The utilities obtained are strictly lower than those obtained in the truthful equilibrium, as can be verified by substituting the equilibrium declarations $x_1, x_2$ into the utility functions. 

Next assume w.l.o.g. that player $2$ has a cost $c_2 \geq a/2$ and $c_2 \geq \frac{1}{3}(2c_1 + a)$. 
Assume by way of contradiction that in the equilibrium of the meta-game both players produce positive quantities. 
On the one hand, the sufficient and necessary condition for positive production by both agents under a declaration profile $(x_1, x_2)$ is $a + x_1 - 2 x_2 > 0$ and $a + x_2 - 2 x_1 > 0$. On the other hand, consider the equilibrium given by Lemma \ref{thm:meta-game-both-produce-NE}. If the equilibrium is $x_1 = \frac{1}{5}(8c_1 - 2c_2 - a)$, $x_2 = \frac{1}{5}(8c_2 - 2c_1 - a)$, by substituting these declarations in the positive-production condition we obtain $a + 2c_1 - 3c_2 > 0$ and $a + 2c_2 - 3c_1 > 0$, which is in contradiction to $c_2 \geq a/2$ and $c_2 \geq \frac{1}{3}(2c_1 + a)$. The other alternative is that the equilibrium is 
$x_1 = 0, x_2 = \frac{1}{4}(6c_2-a)$. Substituting these declarations into the positive-production condition we obtain $c_2 < a/2$, which is a contradiction. 

Therefore, it cannot be that both agents produce. Notice that player $1$ who has a low cost will prefer to produce even if player $2$ declares zero. Therefore, in the equilibrium of the meta-game, player $1$ produces alone, and declares the value that maximizes his utility as a monopolist, subject to the constraint that player $2$ still prefers not to produce, which is $x_1 = 2c_2 - a < c_1$ (since in the truthful Nash equilibrium both players produce, which implies $a + c_1 - 2c_2 > 0$). The best-reply declaration of player $2$ is then any declaration $x_2 \geq c_2$. 

The quantity that is produced by player $1$ in the equilibrium of the meta-game where he produces alone is $\frac{a-c_2}{b}$, which is more than the total quantity of $\frac{1}{3b}(a - c_1 - c_2)$ that is produced in the Nash equilibrium of the game with the true parameters (since $\frac{1}{3b}(a - c_1 - c_2) < \frac{1}{3b}(a - (2c_2-a) - c_2) = \frac{1}{3b}(2a-3c_2) < \frac{a-c_2}{b}$), thus yielding a lower price. The utility of player $1$ in this meta-game equilibrium is $u_1 = \frac{1}{b}(a-c-2)(c_2 - c_1)$. The utility of player $1$ in the Nash equilibrium of the game with the true parameters is $u_1^{true} = \frac{1}{9b}(a + c_2 - 2c_1)^2$. The utility difference $u_1^{true} - u_1 = \frac{1}{9b}\left(a^2 + 10c_2^2 + 4c_1^2+5ac_1 - 7ac_2 - 13c_1c_2\right)$ is negative for $c_2 \geq a/2$ and $c_2 \geq \frac{1}{3}(2c_1 + a)$. That is, the utility of player $1$ who has a low production cost is higher in the equilibrium of the meta-game in which he ``drives player $2$ out of the market'' than in the equilibrium of the game with the true parameters. 
\end{proof}

\vspace{5pt}
Theorem \ref{thm:cournot-manipulation-free} describes the limited set of cases where the game is manipulation-free. Basically, to establish conditions on the game parameters such that the game is manipulation-free, we require that in the equilibrium of the meta-game $x_i = c_i$. The proof is then based on the characterization of the equilibria as specified in Theorem 2 and Lemma \ref{thm:meta-game-both-produce-NE}.  

\begin{proof} (Theorem \ref{thm:cournot-manipulation-free}):  
The game is manipulation free if in the equilibrium of the meta-game $x_1 = c_1$ and $x_2 = c_2$. If in the equilibrium of the meta-game both players produce positive quantities, then by Lemma \ref{thm:meta-game-both-produce-NE} it holds that the Nash equilibrium is $(x_1 = \frac{1}{5}(8c_1 - 2c_2 - a), x_2 = \frac{1}{5}(8c_2 - 2c_1 - a))$ if $x_1^*, x_2^* \geq 0$, $(x_1 = \frac{1}{4}(6c_1-a)^+, x_2 = 0)$ if $x_2^* < 0, x_1^*\geq 0$, $(x_1 = 0, x_2 = \frac{1}{4}(6c_2-a)^+)$ if $x_1^* < 0, x_2^* \geq 0$, and $(x_1=0,x_2=0)$ otherwise. In the first case, where $x_1 = \frac{1}{5}(8c_1 - 2c_2 - a), x_2 = \frac{1}{5}(8c_2 - 2c_1 - a)$, requiring $x_1 = c_1$ and $x_2 = c_2$ yields $x_1=x_2 = a$, in which case the players do not produce in the truthful equilibrium. That is, if both players prefer not to produce at all under the true parameters, then they also do not have any profitable manipulation. In the second case, $x_1 = \frac{1}{4}(6c_1-a)^+, x_2 = 0$, requiring $x_1 = c_1$ and $x_2 = c_2$ yields either $c_1 = a/2$ and $c_2 = 0$, in which case player $1$ does not produce, or $x_1=x_2=0$. Similarly, the symmetric case $x_1 = 0, x_2 = \frac{1}{4}(6c_2-a)^+$ yields either $c_1 = 0$ and $c_2 = a/2$ or $x_1=x_2=0$. Finally, in the case where $x_1^*,x_2^* < 0$, the equilibrium is $x_1=x_2=0$ which is truthful if and only if $c_1=c_2=0$. 

The remaining cases are the case where under the true parameters both players produce, but in the equilibrium of the meta-game only one player produces, and the case where under the true parameters only one player produces. In the first case, as shown in Theorem \ref{thm:cournot-meta-game-NE}, it must be that the player with the lower cost shades his declaration, and so this equilibrium is not truthful. In the other case, we assume w.l.o.g. that player $1$ is a monopolist. Player $1$ then maximizes his utility when declaring his true cost $x_1 = c_1$, producing a  quantity $q^*$. Player $2$ then prefers not to produce. It is not difficult to see that given that player $1$ declares the truth, any declaration of player $2$ that will lead him to produce a positive quantity will result in a total production grater than $q^*=\frac{a-c_1}{2b}$, and thus will give him a negative utility. Thus, player $2$ prefers to declare the truth, $x_2 = c_2$, and the truthful declarations form a Nash equilibrium of the meta-game. 
\end{proof}

\section{Opposing-Interests Games}\label{sec:appendix-opposing-interests-games}

\begin{figure}[!t]
\vspace{-8pt}
\centering
		\hspace{-27pt}
		\includegraphics[width=0.25\linewidth]{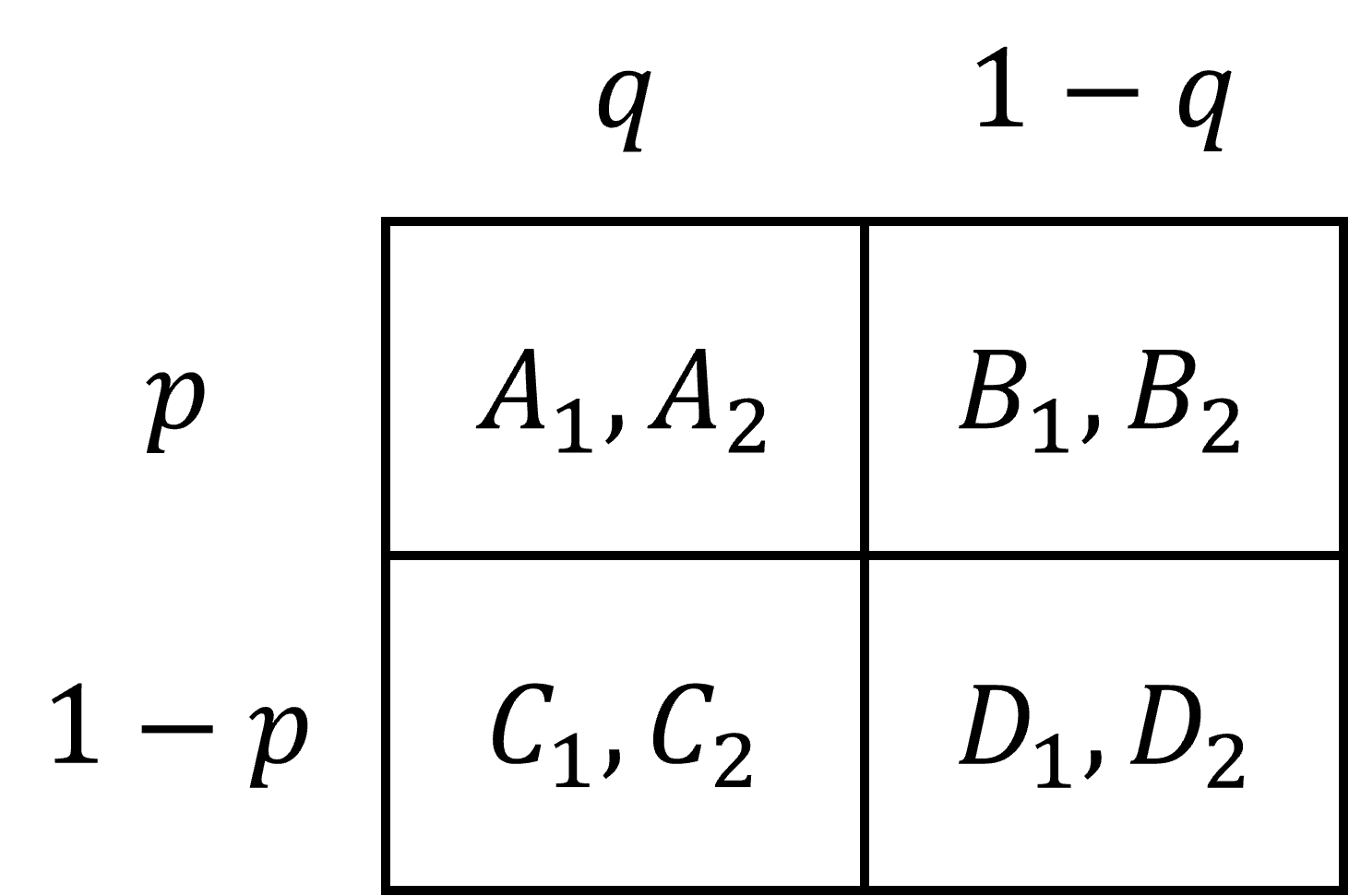}
		\caption{Notation for generic $2$$\times$$2$ games.}  
		\label{fig:2x2-game-notations}
\end{figure}

To set notations, Figure \ref{fig:2x2-game-notations} denotes the utilities of the players in each game outcome where capital letters denote the (pure) game outcomes and the subscripts denote player indices. Thus, in opposing-interests $2$$\times$$2$ games $A_1, D_1 > B_1, C_1$ and $A_2, D_2 < B_2, C_2$, or vice versa, i.e., all inequalities are reversed. Additionally, the parameter $p$ in the figure denotes the probability that the row player plays the top row in a mixed Nash equilibrium and the parameter $q$ denotes the probability that the column player plays the left column.

We start with two technical observations that will be useful for our proofs. These observations are formalizations of standard calculations according to the notation presented in Figure \ref{fig:2x2-game-notations} above. The first observation specifies the utilities of the players in any mixed strategy profile $(p,q)$, and the second observation shows the expressions for $(p,q)$ in a completely mixed Nash equilibrium, when such exists in the game. The proofs of these observations are straightforward.

\begin{observation}\label{mixed_strategy_payoffs} 
Consider a $2$$\times$$2$ game as presented in Figure \ref{fig:2x2-game-notations}. If the players play a mixed strategy profile $(p,q)$, i.e., both players mix their pure strategies such that the row player plays the top row w.p. $p\in[0,1]$ and the column player plays the left column w.p. $q\in[0,1]$, then the expected utilities $u_1$ of the row player and $u_2$ of the column player are given by  
{\small
$
\ u_1 = pq(A_1 - C_1 + D_1 - B_1) + p (B_1 - D_1) + q (C_1 - D_1) + D_1 
$
}
and 
{\small
$
\ u_2 = pq(A_2 - B_2 + D_2 - C_2) + p (B_2 - D_2) + q (C_2 - D_2) + D_2. 
$
}
\end{observation}

\begin{observation}\label{thm:mixed_NE_profile}
The mixed Nash equilibrium profile $(p,q)$ of a $2$$\times$$2$ game (as presented in Figure \ref{fig:2x2-game-notations}), if such exists, is given by 
{\small
$p = \frac{D_2 - C_2}{A_2 - B_2 + D_2 - C_2}$,  $q = \frac{D_1 - B_1}{A_1 - C_1 + D_1 - B_1}$.
}
\end{observation}

Theorem \ref{thm:opposing-interests-NE} characterizes the equilibrium of the meta-game and the utilities obtained in it, when the players use any types of regret-minimizing agents and parameter spaces where any one of the four utilities of each player in the game is a parameter that the player can manipulate. The proof of the theorem first shows that both players will strictly prefer to use declarations that lead to a (``manipulated'') game between the agents that is a fully mixed game, and then the proof characterizes the declaration profiles that form equilibria of the meta-game and shows that all such equilibria lead to the same utilities as those obtained in the truthful Nash equilibrium. The theorem technically holds not only for manipulations of any one of the parameters of each player but for a broader range of parameter spaces which we term ``natural parameter spaces'' that have the following properties.

\begin{definition} A parameter space $P \subseteq \mathcal{R}^4$ of the row player in a $2$$\times$$2$ game with (true) row-player parameters $(A_1,B_1,C_1,D_1)$ is called \emph{natural} if it has the following properties. 
\begin{enumerate}[leftmargin=*]
	\item (Sufficient generality) The parameter space induces all possible values on the column's mixed strategy; i.e.,
for every $0 < q < 1$, there exists a declaration profile $(A'_1,B'_1,C'_1,D'_1) \in P$ such that $q \cdot A'_1 + (1-q) \cdot B'_1 = q \cdot C'_1 + (1-q) \cdot D_1'$.
	
	\item (Analyzability)  For every $(A'_1,B'_1,C'_1,D'_1) \in P$ either the best-replies to pure strategies do not change, i.e., $sign(A_1-C_1)=sign(A'_1-C'_1)$, and $sign(B_1-D_1)=sign(B'_1-D'_1)$, or
the player has a dominant strategy, i.e., $sign(A'_1-C'_1)=sign(B'_1-D'_1)\neq 0$.
\end{enumerate}
where $sign(x) = 1$ if $x>0$, $sign(x) = -1$ if $x<0$, and $sign(x) = 0$ if $x=0$.

\vspace{5pt}
\noindent
The parameter space for the column player is called \emph{natural} in a similar manner, and the whole parameter space of the game is called \emph{natural} if it is natural for both players.
\end{definition}

The following lemmas show examples of natural parameter spaces for opposing-interests games. The first example describes user manipulations that include all or any subset of the agent's utilities, as long as the best-reply structure of the game is preserved, and the second example includes arbitrary manipulations of any single parameter by each user. 
\begin{lemma}\label{thm:four-param-natural-parameter-space} 
For every opposing-interests game and every choice of a non-empty subset of the four parameters for each player, the parameter space where the non-chosen parameters are fixed to their true values and the chosen parameters are arbitrary as long as they conserve best-replies to pure strategies (i.e., for the row player  $sign(A_1-C_1)=sign(A'_1-C'_1)$ and $sign(B_1-D_1)=sign(B'_1-D'_1)$, and for the column player $sign(A_2-B_2)=sign(A'_2-B'_2)$ and $sign(C_2-D_2)=sign(C'_2-D'_2)$) is a natural parameter space. 
\end{lemma} 
\begin{proof}  
Property (2) of a natural parameter space holds immediately in this case. To show property (1), let $q\in(0,1)$, and we require that for some $(A'_1,B'_1,C'_1,D'_1) \in P_1$ it holds that $q \cdot A'_1 + (1-q) \cdot B'_1 = q \cdot C'_1 + (1-q) \cdot D_1'$, and similarly for the column player let $p\in(0,1)$, and we require that for some $(A'_2,B'_2,C'_2,D'_2) \in P_2$ it holds that $p \cdot A'_2 + (1-p) \cdot C'_1 = p \cdot B'_1 + (1-p) \cdot D_2'$, yielding 
\begin{equation*}
p = \frac{1}{1 + \frac{A_2' - B_2'}{D_2' - C_2'}}, \quad \quad \quad q = \frac{1}{1 + \frac{A_1' - C_1'}{D_1' - B_1'}}.
\end{equation*}
Notice, that under the conditions of the lemma the values of these two expressions lie in the range $(0,1)$ since both terms in the denominators are positive. 
We next solve for {\small  $\frac{A_2' - B_2'}{D_2' - C_2'}$} and for {\small$\frac{A_1' - C_1'}{D_1' - B_1'}$}: 
\small   
\begin{equation*}
\frac{A_1' - C_1'}{D_1' - B_1'} = \frac{1-q}{q}, \quad \quad \quad \frac{A_2' - B_2'}{D_2' - C_2'} = \frac{1-p}{p}.
\end{equation*}
\normalsize

These equations can be easily satisfied by selecting declarations within the parameter space of each player, since only a single degree of freedom is required to do so. To demonstrate this for the row player (equation for $q$), assume w.l.o.g. that $A_1>C_1$ and $D_1>B_1$, denote $y = (1-q)/q >0$, and denote a single free parameter of this player by $x$, and assume that the other parameters are the true parameters of the game. We will look at the four cases where $x$ replaces each one of the parameters of the true game. If $x$ replaces $A_1$, then declaration $x = C_1 + (D_1-B_1)y$, which is in $P_1$ (since $x>C_1$, and thus it preserves best-replies), satisfies the equation for $q$. If $x$ replaces $C_1$, then declaration $x = A_1 - (D_1-B_1)y$, which is in $P_1$ (since $x<A_1$), satisfies the equation for $q$. If $x$ replaces $D_1$, then declaration $x = B_1 + (A_1-C_1)/y$, which is in $P_1$ (since $x>B_1$), satisfies the equation for $q$. If $x$ replaces $B_1$, then declaration $x = D_1 - (A_1-C_1)/y$, which is in $P_1$ (since $x<D_1$), satisfies the equation for $q$. Similar considerations apply to the declarations of the column player and the equation of $p$. Additional free parameters also allow the row player to induce $q$ as a mixed strategy (e.g., in a degenerate way of controlling only one parameter), and allow to the column player to induce $p$. 
\end{proof}
\begin{lemma}\label{thm:single_param_natural_parameter_space} 
For any opposing-interests game and choice of one of the parameters of each player, the parameter space where three parameters are fixed to the true values and the fourth parameter is arbitrary (except for exactly being equal to another parameter) is a natural parameter space.
\end{lemma}
\begin{proof}  
Consider any opposing-interests $2$$\times$$2$ game with utilities $(A_1,B_1,C_1,D_1)$ to the row player and $(A_2,B_2,C_2,D_2)$ to the column player (see Figure \ref{fig:2x2-game-notations}), and assume that each player can set the declaration of one of his parameters to any arbitrary value with generic parameters, i.e., without equalities in utilities between actions. 

Property (1) of a natural parameter space can be obtained by the same argument given in the proof of Lemma \ref{thm:four-param-natural-parameter-space}   
above. Regarding property (2), for the case that the player does not change his best-replies to pure strategies in his declaration, property (2) holds directly. 
If a player does declare a single parameter that changes his best-reply to a pure strategy, then since the game is an opposing-interests game 
and only a single parameter has changed, this implies that this  player now has a dominant strategy, and so property (2) holds in this case as well.
\end{proof}

To prove Theorem \ref{thm:opposing-interests-NE} we consider the following restatement of the theorem. 
\begin{theorem*} (Restatement of Theorem \ref{thm:opposing-interests-NE}):  
In any $2$$\times$$2$ game with opposing interests and a natural parameter space, the Nash equilibrium of the meta-game is essentially unique. That is, 
there is a unique strategy profile $(p,q)\in(0,1)^2$ such that every Nash equilibrium of the meta-game induces $(p,q)$ as a unique Nash equilibrium of the agents' game. The utility of each player in that equilibrium is equal to the utility of the player when all players enter their true parameters into their agents.
\end{theorem*}\label{thm:opposing_interests_game_unique_meta_game_NE}
\begin{proof}
Consider any opposing-interests $2$$\times$$2$  game with a natural parameter space and with true utilities $(A_1,B_1,C_1,D_1)$ of the row player and $(A_2,B_2,C_2,D_2)$ of the column player (see Figure \ref{fig:2x2-game-notations}) and assume w.l.o.g. that $A_1>B_1,C_1$ and $D_1>B_1,C_1$ and $B_2>A_2,D_2$ and $C_2>A_2,D_2$, i.e., that the row player has higher utilities on the main diagonal of the game utility matrix and the column player has higher utilities off the diagonal.  

With a natural parameter space the parameter declarations lead to one of the following cases: (a) no player changes the signs of his best-replies, or (b) one or both players declare parameters such that their agents have dominant strategies in the game with the declared parameters. We argue that (b) cannot be the case in a Nash equilibrium of the meta-game. If, for example, the row player declares a dominant strategy to his agent, say to play the top row, then the best-reply of the column player would be to declare the right column as a dominant strategy to his agent. However, this is not a Nash equilibrium of the meta-game since in this declaration profile the row player has utility $B_1$ and so he would prefer to change strategy to any mixed strategy $0 \leq p < 1$ to obtain utility $p B_1 + (1-p)D_1 > B_1$. A similar argument holds for any other dominant strategy declaration. 

Thus, in a Nash equilibrium of the meta-game the players do not declare dominant strategies for their agents; i.e., the strategy of the row player in a Nash equilibrium of the meta-game is to declare parameters that induce $q\in(0,1)$, and the strategy of the column player in a Nash equilibrium of the meta-game is to declare parameters that induce $p\in(0,1)$. This also implies that in a Nash equilibrium of the meta-game no player reverses any of the directions of his best-replies in the game, and therefore the game with the declared parameters still has a single mixed Nash equilibrium to which the regret-minimizing agents converge, and so the utilities of the players in the meta-game can be analytically derived. The utilities of the two players when the agents converge to such a mixed strategy profile $(p,q)$ are, by Observation \ref{mixed_strategy_payoffs},
\small
\begin{equation*}
u_1 = pq(A_1 - C_1 + D_1 - B_1) + p (B_1 - D_1) + q (C_1 - D_1) + D_1
\end{equation*}
\begin{equation*} 
u_2 = pq(A_2 - B_2 + D_2 - C_2) + p (B_2 - D_2) + q (C_2 - D_2) + D_2. 
\end{equation*}
\normalsize
Since $p$ is a function of only the column player's (player 2's) declared parameters and the true game parameters, and $q$ is  a function of only the row player's (player 1's) parameters and the true game parameters, and in a natural parameter space the row player can choose parameters to induce any $q \in (0,1)$, and the column player can similarly choose parameters to induce any  $p\in(0,1)$, we can think of $q$ of the Nash equilibrium of the agents' game as (essentially) the strategy of the row player, and similarly, of $p$ as (essentially) the strategy of the column player. Thus, the condition for a Nash equilibrium in the meta-game is {\small $\frac{\partial u_1}{\partial q} = \frac{\partial u_2}{\partial p}$ = 0}.
  
Hence, the strategy profile $(p,q)$ in any Nash equilibrium of the meta-game is 
\small
\begin{equation*}
p = \frac{D_1 - C_1}{A_1 - B_1 + D_1 - C_1}, \quad \quad \quad q = \frac{D_2 - B_2}{A_2 - B_2 + D_2 - C_2}.
\end{equation*}
\normalsize

Next, we calculate the utilities in an equilibrium of the meta-game and compare them with the utilities of the Nash equilibrium of the game with the true parameters. By substituting the unique $(p,q)$ Nash equilibrium  profile of the meta-game into the equations of the utilities shown above, we get that the utilities in any Nash equilibrium of the meta-game are 
\small
\begin{equation*}
u_1 = \frac{(D_1 - C_1)(B_1-D_1)}{A_1 - C_1 + D_1 - B_1} + D_1, \quad \quad \quad u_2 = \frac{(D_2 - C_2)(B_2 - D_2)}{A_2 - B_2 + D_2 - C_2} + D_2.
\end{equation*}
\normalsize
The Nash equilibrium strategy profile of the game with the true parameters is (by Observation \ref{thm:mixed_NE_profile})
\small
\begin{equation*}
p^{NE} = \frac{D_2 - C_2}{A_2 - B_2 + D_2 - C_2}, \quad \quad \quad q^{NE}= \frac{D_1 - B_1}{A_1 - C_1 + D_1 - B_1}, 
\end{equation*}
\normalsize 
and the Nash equilibrium utilities of the game with the true parameters are 
\small 
\begin{equation*}
u_1^{NE} = \frac{(D_1 - B_1)(C_1-D_1)}{A_1 - C_1 + D_1 - B_1} + D_1, \quad \quad \quad u_2^{NE} = \frac{(D_2 - C_2)(B_2 - D_2)}{A_2 - B_2 + D_2 - C_2} + D_2.
\end{equation*}
\normalsize
That is, in a Nash equilibrium of the meta-game the players have the same utilities as in the Nash equilibrium of the game with the true declarations. 
\end{proof}

Using Theorem \ref{thm:opposing-interests-NE} we can determine necessary and sufficient conditions for the meta-game to be manipulation-free. Before continuing to the proof of Theorem 5, we consider the following technical restatement of the theorem, where the game utilities are denoted as indicated in Figure \ref{fig:2x2-game-notations}.  

\begin{theorem*}(Restatement of Theorem 5):   
An opposing-interests $2$$\times$$2$ game with a natural parameter space is manipulation-free iff $\frac{D_1 - C_1}{A_1 - B_1 + D_1 - C_1}$$=$$\frac{D_2 - C_2}{A_2 - B_2 + D_2 - C_2}$ and $\frac{D_2 - B_2}{A_2 - B_2 + D_2 - C_2}$$=$$\frac{D_1 - B_1}{A_1 - B_1 + D_1 - C_1}$.
\end{theorem*}
Notice that the above condition exactly specifies that the Nash equilibrium $(p,q)$ of the game (see Observation \ref{thm:mixed_NE_profile}) is symmetric to permutations of player indices, as stated in the theorem. 
\begin{proof} 
A game is manipulation-free if the mixed strategy profile of the agents in a Nash equilibrium of the meta-game is identical to the Nash equilibrium profile of the game with the true parameters, where we used the fact shown in the proof of Theorem \ref{thm:opposing-interests-NE} above, that any opposing-interests $2$$\times$$2$ game with a natural parameter space has a unique mixed strategy profile $(p,q)\in(0,1)^2$ that is induced by every Nash equilibrium of the meta-game, and the fact that opposing-interests games have a unique mixed Nash equilibrium.\\   

Using Observation \ref{thm:mixed_NE_profile}, the (true-parameters) Nash equilibrium profile is 
{\small
$p^* = \frac{D_2 - C_2}{A_2 - B_2 + D_2 - C_2}$ and $q^*= \frac{D_1 - B_1}{A_1 - C_1 + D_1 - B_1}$
}. 
Next, using Observation \ref{mixed_strategy_payoffs}, we can write the utilities of the players as functions of $p,q$, as follows. 
{\small
\begin{equation*}
u_1 = pq(A_1 - C_1 + D_1 - B_1) + p (B_1 - D_1) + q (C_1 - D_1) + D_1, 
\end{equation*}
\begin{equation*} 
u_2 = pq(A_2 - B_2 + D_2 - C_2) + p (B_2 - D_2) + q (C_2 - D_2) + D_2. 
\end{equation*}
}

As in the proof of Theorem \ref{thm:opposing-interests-NE} above, since $p$ is a function of only the column player's (player 2's) declared parameters and the true game parameters, and $q$ is a function of only the row player's (player 1's) parameters and the true game parameters, and in a natural parameter space the row player can choose parameters to induce any $q \in (0,1)$, and the column player can similarly choose parameters to induce any  $p\in(0,1)$, we can think of $q$ of the Nash equilibrium of the agents' game as the strategy of the row player, and similarly, of $p$ as the strategy of the column player, and  thus the condition for a Nash equilibrium in the meta-game is {\small $\frac{\partial u_1}{\partial q} = \frac{\partial u_2}{\partial p}$ = 0}.

Hence, the strategy profile $(p,q)$ in a Nash equilibrium of the meta-game is
{\small
$\tilde{p} = \frac{D_1 - C_1}{A_1 - B_1 + D_1 - C_1}$, $\tilde{q} = \frac{D_2 - B_2}{A_2 - B_2 + D_2 - C_2}$.
}
Requiring $p^*=\tilde{p}$ and $q^*=\tilde{q}$, we obtain the condition stated in the theorem. 
\end{proof}

\subsection{Opposing-interests game example}\label{appendix:example}

Here we provide further details on the example presented in Section \ref{c2d3}. 
 
\vspace{5pt}
\noindent
\textbf{Nash equilibrium of the (true) game:} 
Using Observation \ref{thm:mixed_NE_profile} and the parameters of the game shown in Figure \ref{fig:fully_mixed_manipulation_game_matrix} (left) in the main text, we obtain that the Nash equilibrium of the game is the mixed strategy profile: $p=2/3$ and $q=2/5$. Using Observation \ref{mixed_strategy_payoffs}, we obtain that the utilities in this mixed strategy profile are $u_1=1/5$ for the row player and $u_2=1/3$ for the column player.

\vspace{5pt}
\noindent
\textbf{Nash equilibrium of the manipulated game:} For any declarations $c,d>-1$ of the players, the ``manipulated game'', i.e., the game with the declared parameters that is played by the agents, has a unique mixed Nash equilibrium, with mixed strategies (using Observation \ref{thm:mixed_NE_profile}) $p = \frac{d+1}{d+3}$ and $q = \frac{2}{c+3}$. The true expected utilities of the users can be calculated with these values of $p$ and $q$ and the true payoffs of the game 
by using Observation \ref{mixed_strategy_payoffs}, yielding 
\vspace{-5pt} 
\small
\begin{equation*}
u_1 = 5pq -2p -2q + 1 = \frac{c (1-d) + 3d + 1}{(c+3)(d+3)}, \ \ \ u_2 = 4q + 2p -6pq -1 = \frac{c (d-1) - d + 9}{(c+3)(d+3)}.
\end{equation*}
\normalsize
In this example, when the row player declares $c=1$, and the column player declares the truth, $d=3$, the utilities to the two players according to the above expressions are $u_1=1/3$ and $u_2=1/3$. Notice that in this example, this unilateral manipulation increased the utility to the player that manipulated his agent, while the utility to the other player remained the same as in the truthful declarations case. Next, the example describes a unilateral manipulation by the column player, in which the declarations are $c=2$ (i.e., the truthful declaration for the row player) and $d=1$. In this case, the utilities are $u_1=1/5$ and $u_2=2/5$. Here again, a unilateral manipulation by one player increased this player's utility without changing the other player's utility compared with the truthful declarations utilities. Yet, these two declaration profiles described above are not equilibria in the meta-game. 

\vspace{5pt}
\noindent
\textbf{Nash equilibrium of the meta-game:} The equilibrium condition for the meta-game is {\small$\frac{\partial u_1}{\partial c} = \frac{2-6d}{(c+3)^2(d+3)} = 0$}, and {\small$\frac{\partial u_2}{\partial d} = \frac{4(c-3)}{(c+3)(d+3)^2} = 0$}. Hence, the Nash equilibrium of the meta-game is the declaration profile $c=3$, $d=1/3$. As can be seen using Observation \ref{mixed_strategy_payoffs}, the utilities of the players in this distribution of play 
are the same as their utilities in the Nash equilibrium of the (true) game. As discussed in the main text, this result is generalized in 
Theorem \ref{thm:opposing-interests-NE}.

\begin{figure}[!t]
\vspace{6pt}
\centering
	\begin{subfigure}{.41\linewidth} 
	\center
		\includegraphics[width=0.94\linewidth]{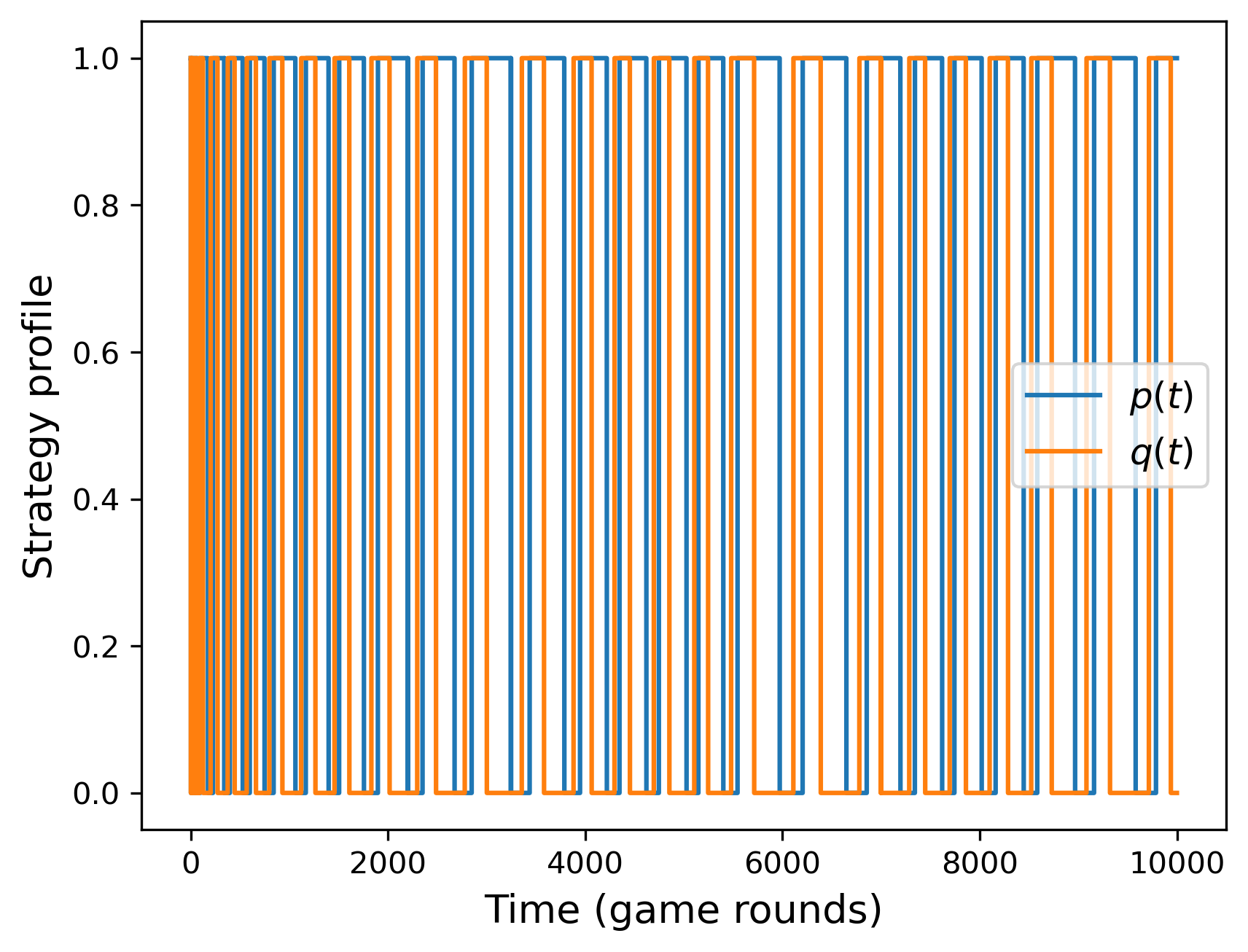}
		\caption{Bid dynamics of FTPL agents}  
		\label{fig:fully_mixed_p_q_dynamics_FTPL}
	\end{subfigure}
	\begin{subfigure}{.54\linewidth}

		\includegraphics[width=1.08\linewidth]{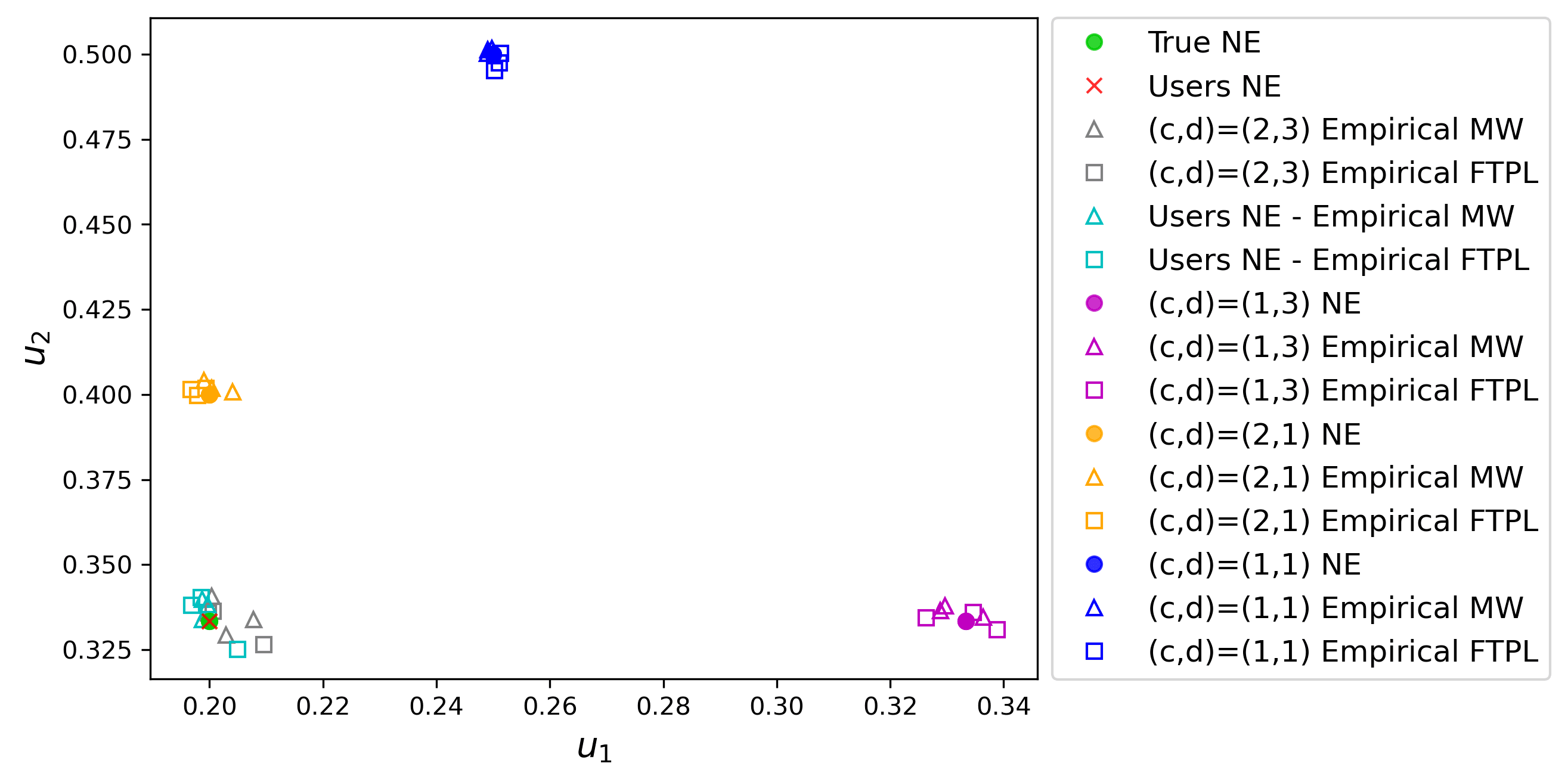}
		\caption{Average user utilities for MW and FTPL agents}  
		\label{fig:fully_mixed_payoffs_A}
	\end{subfigure}
		\hspace{5pt}
\caption{Dynamics and user utilities in the opposing-interests game example from Section \ref{c2d3}, for multiplicative-weights (MW) agents and follow the perturbed leader (FTPL) agents. (a) the dynamics of the mixed strategies $p$ and $q$ of two FTPL agents in $10$,$000$ game rounds. (b) average user utilities in $100$,$000$ game rounds for different manipulation profiles $(c,d)$, where the true values of the utilities of the users are $c=2, d=3$ (as presented in Figure 
\ref{fig:fully_mixed_manipulation_game_matrix} in the main text), compared with the theoretical Nash Equilibrium (NE) utilities. The triangular markers show results of MW agents in three simulation runs and the square markers show results of FTPL agents in three simulation runs. }
\label{fig:FTPL_and_MW}
\end{figure}

\vspace{5pt}
\noindent
\textbf{Additional simulations:} 
Figure \ref{fig:FTPL_and_MW} shows a comparison of simulations of ``follow the perturbed leader'' (FTPL) \cite{kalai2005efficient} with the multiplicative-weights (MW) algorithm \cite{arora2012multiplicative}. Figure \ref{fig:fully_mixed_p_q_dynamics_FTPL} shows an example of the dynamics of two FTPL agents playing the game example presented in Section \ref{c2d3} in the main text, and Figure \ref{fig:fully_mixed_payoffs_A} shows the utilities of the users for several manipulation profiles, including all those presented in the example. The declaration profiles are shown in the legend. It can be seen that the average utilities for the users obtained from the dynamics of the two algorithm types are similar and close to the theoretical Nash equilibrium utilities of each manipulation profile. The manipulation profile $c=d=1$ (marked in blue in the figure) demonstrates an example of a manipulation profile which leads to increased payoffs to both users compared with the truthful declarations (however, it is not an equilibrium). Additional simulations are shown in the following figures. 

Figure \ref{fig:fully_mixed_p_q_dynamics_examples} shows additional examples of the dynamics of multiplicative-weights agents in the same opposing-interests game. Figure \ref{fig:fully_mixed_p_q_dynamics_parametric} shows the learning dynamics in a parametric plot, showing the evolution of joint strategy profiles. Every point is the empirical mixed strategy profile of the agents, presented on the $p,q$ plane, where consecutive points in time are connected by a line. The left panel depicts the dynamics with the parameters used in Figure \ref{fig:fully_mixed_dynamics_and_payoffs_MW} in the main text: $\eta=0.01$ and $T=50$,$000$ game repetitions, and the right panel shows the dynamics with $\eta=0.001$ and $T=1$,$000$,$000$. 

Figure \ref{fig:fully_mixed_game_errors} shows estimates of the deviations of the time average of multiplicative-weights agents dynamics from the Nash equilibrium distribution in the opposing-interests game example. In the left panel, it can be seen that, for a fixed update step size ($\eta=0.01$), indeed these inaccuracies in the convergence of the agents to the Nash distribution (in the time average sense) decrease as $O(1/{\small\sqrt{T}})$, as theoretically expected.  The right panel shows the distribution of the mixed strategies $p,q$ of the two agents for the case of $T=50$,$000$ across $N=1$,$000$ simulation repetitions, which are narrowly centered  near the Nash equilibrium profile.

\clearpage

\begin{figure}[!t]
\centering
	\begin{subfigure}{.245\linewidth} 
	\center
		\includegraphics[width=1.02\linewidth]{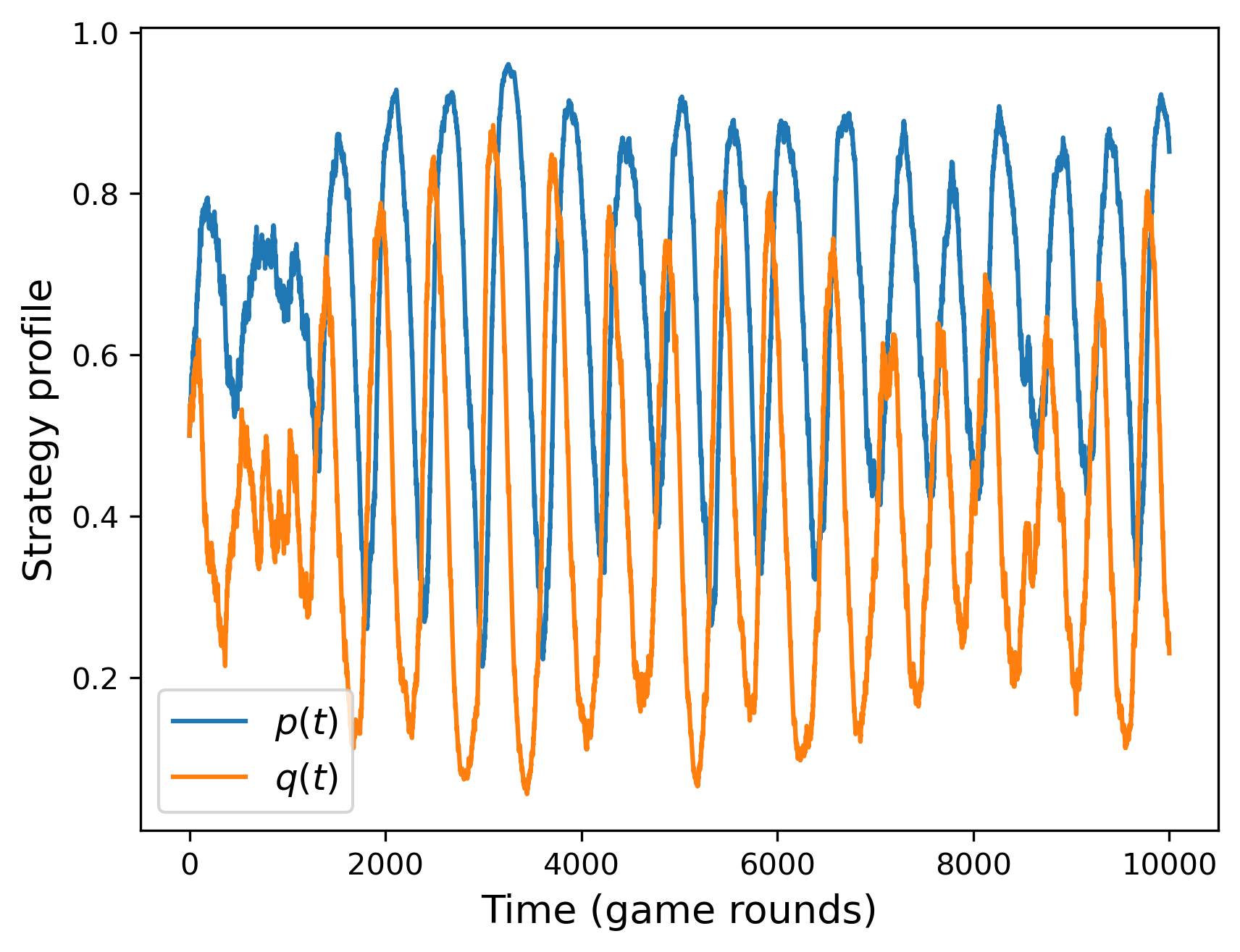}
		\caption{$T=10$,$000$}  
		\label{fig:fully_mixed_p_q_dynamics_10k}
	\end{subfigure}
	%
	\begin{subfigure}{.245\linewidth} 
	\center
		\includegraphics[width=1.02\linewidth]{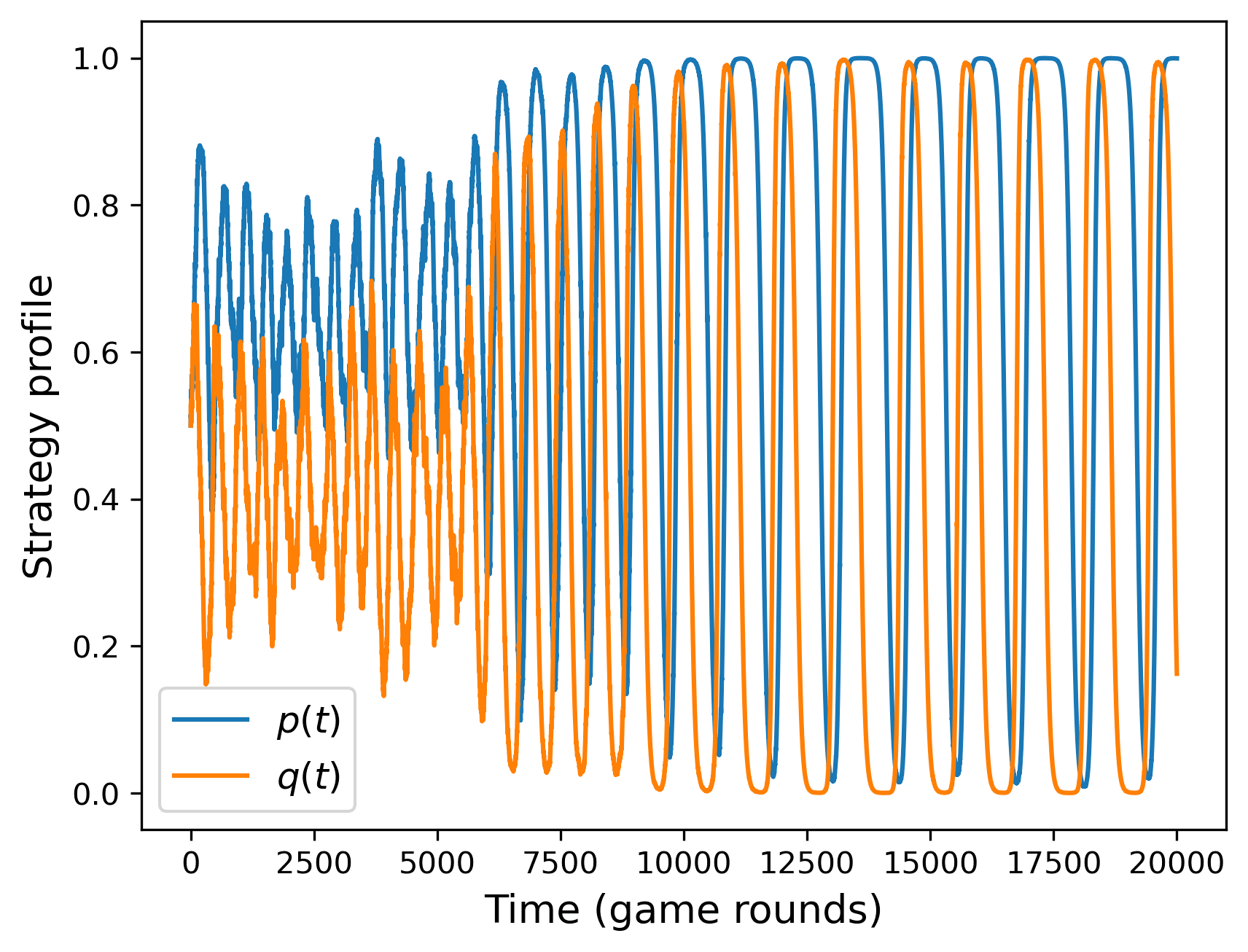}
		\caption{$T=20$,$000$}  
		\label{fig:fully_mixed_p_q_dynamics_20k}
	\end{subfigure}
	%
	%
	\begin{subfigure}{.245\linewidth} 
	\center
		\includegraphics[width=1.02\linewidth]{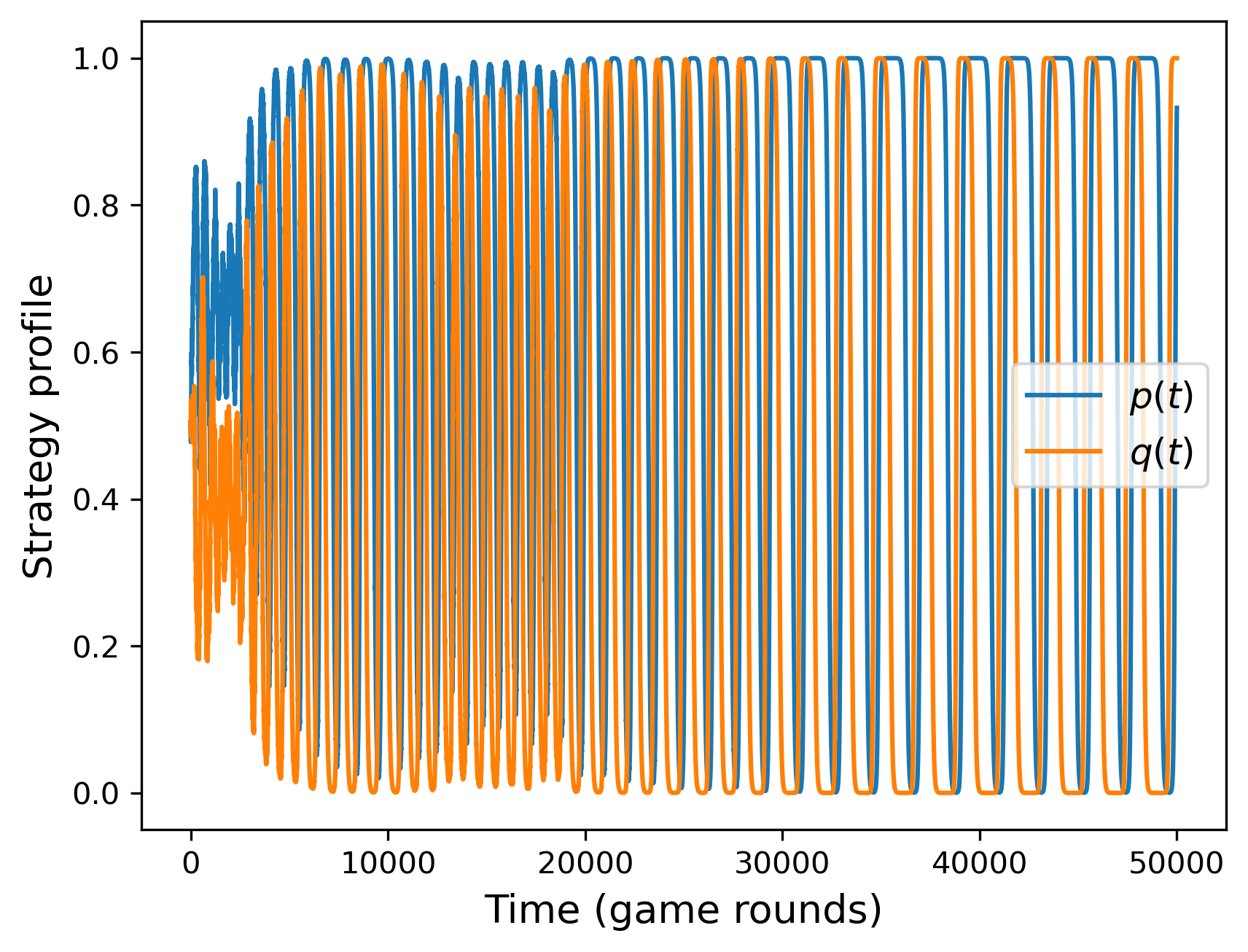}
		\caption{$T=50$,$000$}  
		\label{fig:fully_mixed_p_q_dynamics_50k}
	\end{subfigure}
	%
	%
	\begin{subfigure}{.245\linewidth} 
	\center
		\includegraphics[width=1.02\linewidth]{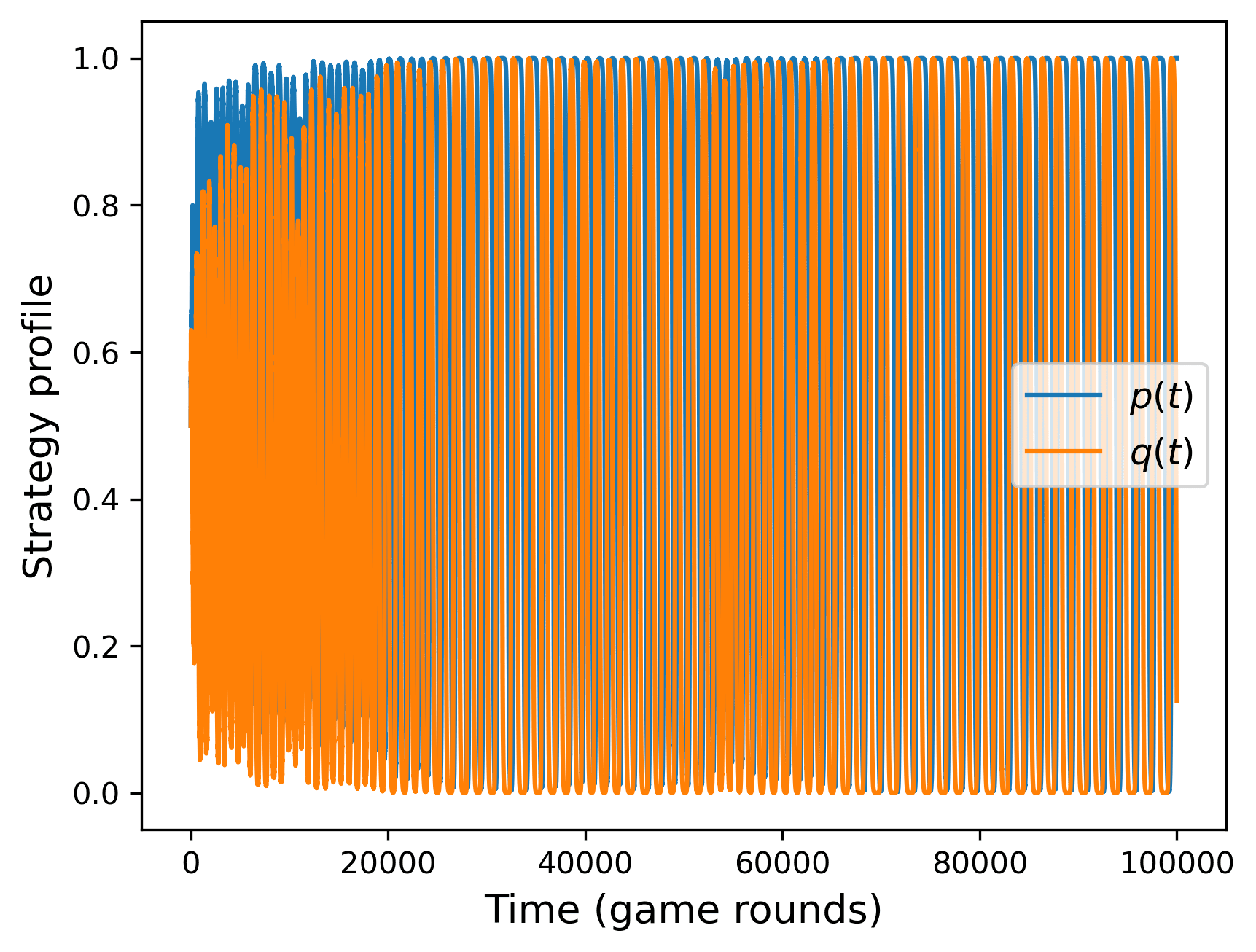}
		\caption{$T=100$,$000$}  
		\label{fig:fully_mixed_p_q_dynamics_100k}
	\end{subfigure}
	\vspace{-2pt}
\caption{
Simulations of multiplicative-weights agents in the opposing-interests game example from Section \ref{c2d3}. The figures show the dynamics of the mixed strategies across game repetitions for different simulation lengths $T$.  
}
\label{fig:fully_mixed_p_q_dynamics_examples}
\end{figure}

\begin{figure}[!ht]
\centering
	\begin{subfigure}{.35\linewidth} 
	\center
		\includegraphics[width=1.0\linewidth]{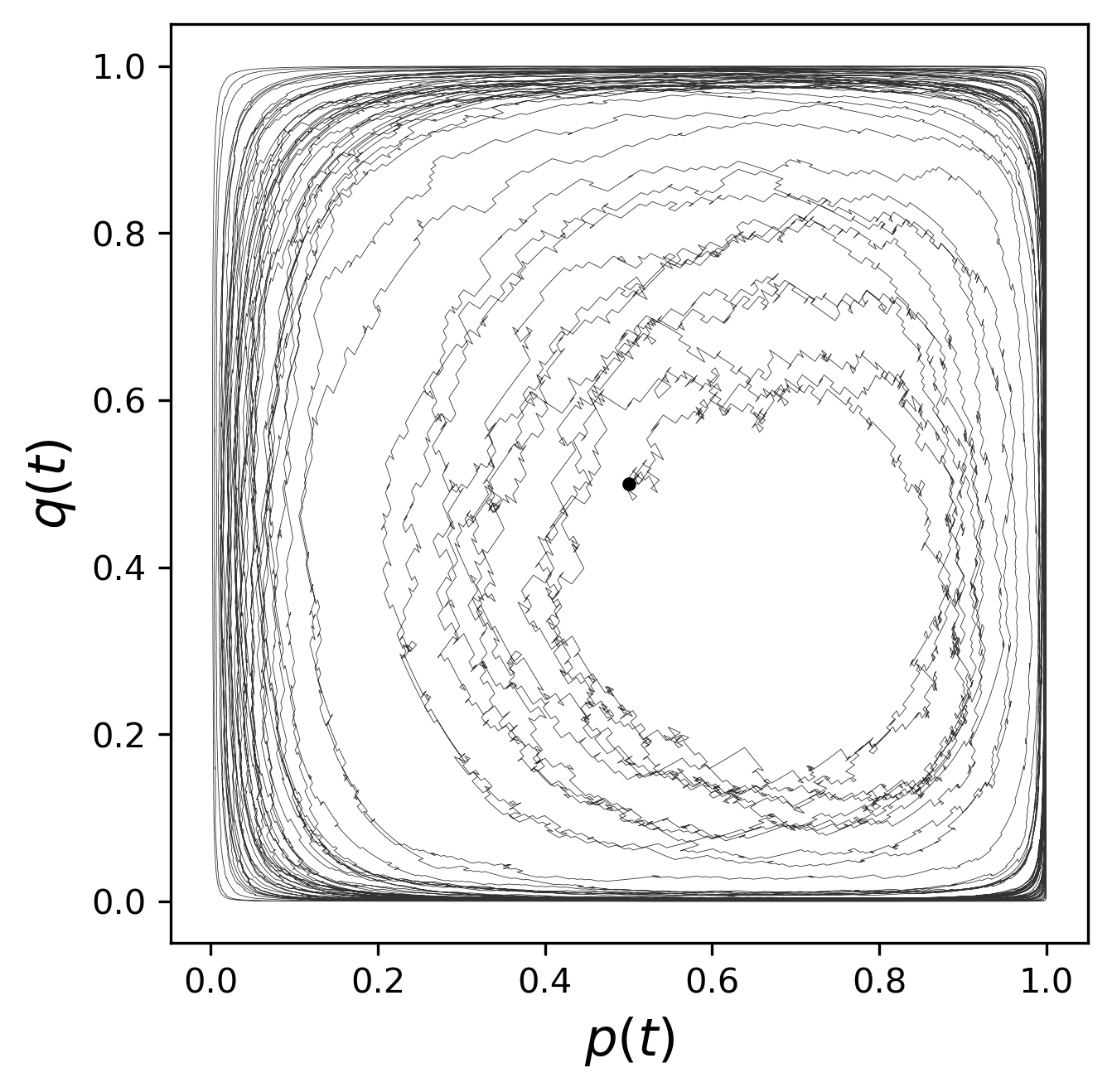}
		\caption{$T=50$,$000$, $\eta=0.01$.}  
		\label{fig:fully_mixed_p_q_parametric_eta=001_T=50K}
	\end{subfigure}
	\hspace{48pt}
	\begin{subfigure}{.35\linewidth} 
	\center
		\includegraphics[width=0.99\linewidth]{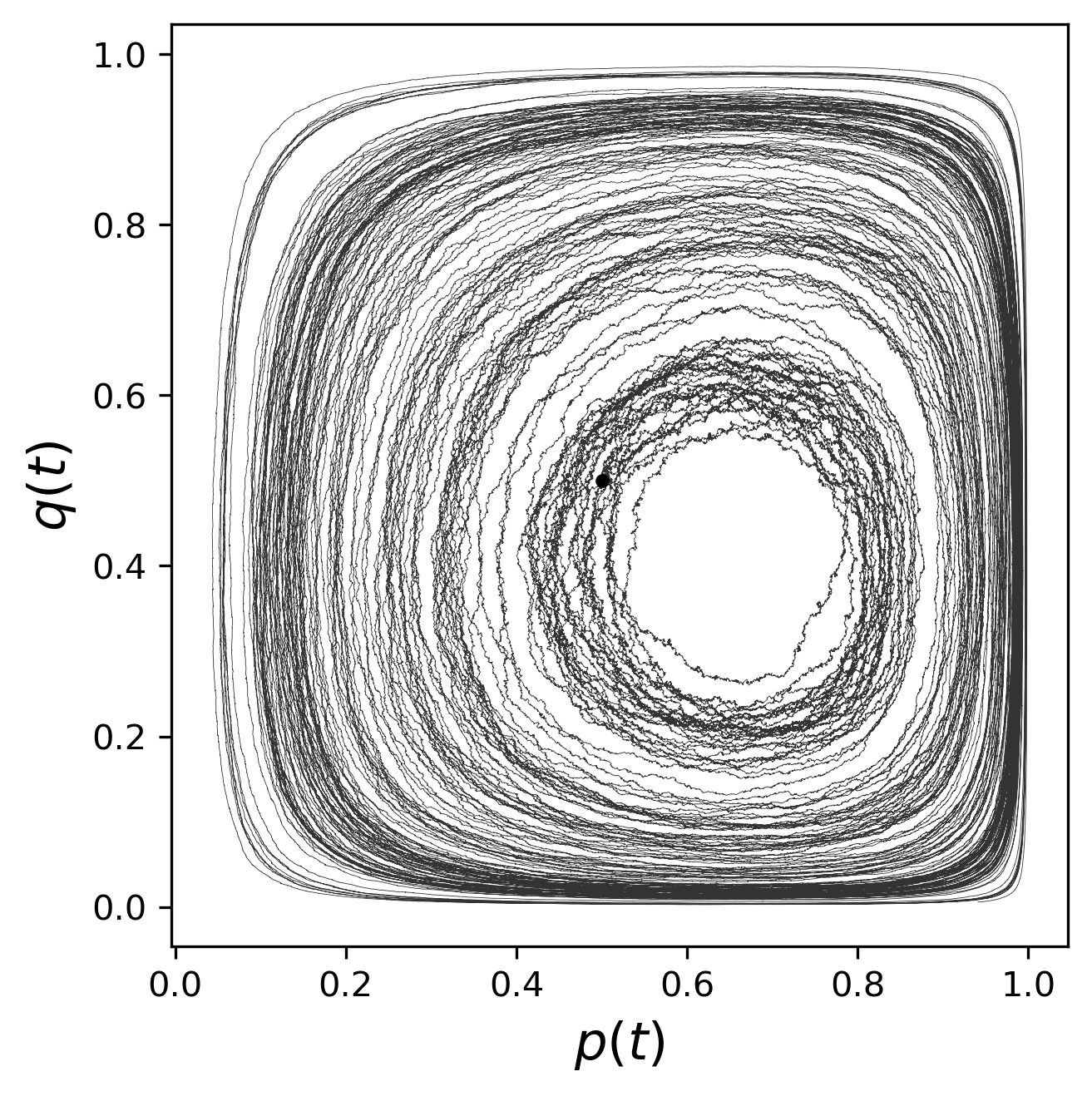}
		\caption{$T=1$,$000$,$000$, $\eta=0.001$.}  
		\label{fig:fully_mixed_p_q_parametric_eta=0001_T=1000K}
	\end{subfigure}
		\hspace{16pt}
\vspace{-5pt}
\caption{
Dynamics of multiplicative-wights agents in the opposing-interests game example from Section \ref{c2d3}. The figures show parametric plots of the mixed strategies of the two agents in time. (a) the dynamics with update step size parameter $\eta=0.01$ and $T=50$,$000$ game repetitions. (b) the dynamics with update step size $\eta=0.001$ and $T=1$,$000$,$000$ game repetitions. 
}
\label{fig:fully_mixed_p_q_dynamics_parametric}
\end{figure}

\begin{figure}[!ht]
\centering
	\begin{subfigure}{.49\linewidth} 
	\center
		\includegraphics[width=0.83\linewidth]{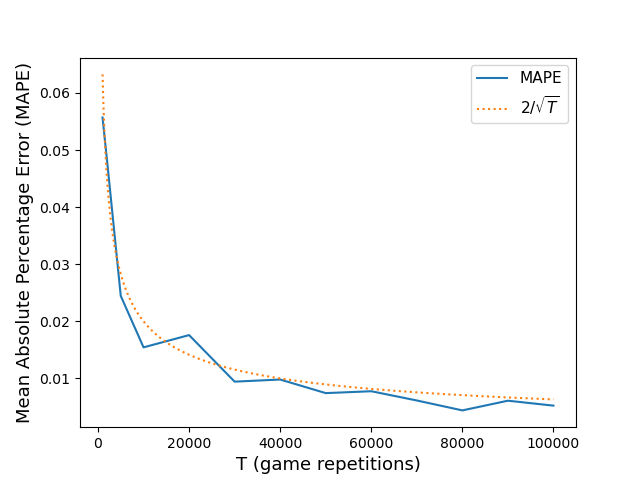}
		\caption{Mean Absolute Percentage Error (MAPE) between the empirical distribution and the NE distribution, obtained in simulations of multiplicative-weights agents playing $T$ game rounds, as a function of $T$. The full line shows the average MAPE and the dotted line shows a comparison to the function $2/{\small\sqrt{T}}$, showing that the empirical error indeed decreases as $O(1/{\small \sqrt{T}})$.}  
		\label{fig:fully_mixed_game_MAPE_vs_T}
	\end{subfigure}
	 \hspace{1pt}
	\begin{subfigure}{.49\linewidth} 
	\center
		\includegraphics[width=0.82\linewidth]{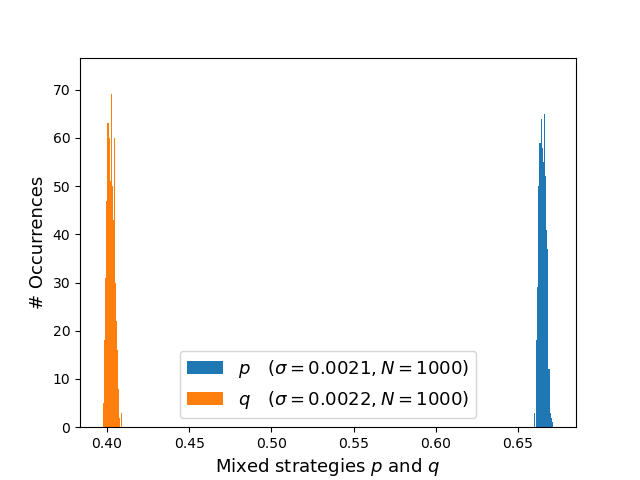}
		\vspace{-2pt}
		\caption{Histograms of the empirical average values of the action probabilities of multiplicative-weights agents in $T=50$,$000$ game rounds. The histogram of $p$ (the probability that the row agent plays the top row) is shown in blue, and of $q$ (the probability that the column agent plays the left column) is shown in orange. The legend shows the sample sizes and standard deviations. }  
		\label{fig:fully_mixed_game_p_q_histogram}
	\end{subfigure}
\caption{Estimates of the deviation of multiplicative-weights agents' time average frequency of play from convergence to the Nash equilibrium distribution in the opposing-interests game example from Section \ref{c2d3}.  
}
\label{fig:fully_mixed_game_errors}
\end{figure}


\end{document}